\documentclass[11pt,draft]{article}

\usepackage{amsmath,amssymb,amsthm}
\usepackage[active]{srcltx}

\usepackage{cite}

\newcommand{\R}{\mathbb{R}}
\newcommand{\C}{\mathbb{C}}
\newcommand{\G}{{\cal G}}
\newcommand{\K}{{\cal K}}
\newcommand{\F}{{\cal F}}

\newcommand{\pr}{{\rm pr}}

\newcommand{\lr}{\lrcorner\,}

\newcommand{\B}{{\cal B}}

\renewcommand{\S}{{\cal S}}

\newcommand{\h}{{\cal H}}

\newcommand{\eps}{\epsilon}
\newcommand{\supp}{{\rm supp}}

\newcommand{\Cyl}{{\rm Cyl}}

\newcommand{\bld}[1]{\boldsymbol{#1}}
\newcommand{\tl}[1]{\tilde{#1}}
\newcommand{\tll}{\tilde{l}}
\newcommand{\we}{\wedge}

\newcommand{\la}{\lambda}
\newcommand{\La}{\Lambda}

\newcounter{mnotecount}[section]

\newtheorem{thr}{Theorem}
\newtheorem{lm}[thr]{Lemma}

\newtheorem{pro}[thr]{Proposition}

\numberwithin{equation}{section}
\numberwithin{thr}{section}

\begin{document}

\title{Constrained projective quantum states for the degenerate Pleba\'nski gravity}

\author{Andrzej Oko{\l}\'ow}

%\date{February 10, 2017}
\date{March 21, 2017}

\maketitle
\begin{center}
{\it  Institute of Theoretical Physics, University of Warsaw\\ ul. Pasteura 5, 02-093 Warszawa, Poland\smallskip \smallskip\\
oko@fuw.edu.pl}
\end{center}
\medskip

\begin{abstract}
Nowadays projective quantum states can be constructed for a number of field theories including Loop Quantum Gravity. However, these states are kinematic in this sense that their construction does not take into account the dynamics of the theories. In particular, the construction neglects constraints on phase spaces. Here we present projective quantum states for a ``toy-theory'' called degenerate Pleba\'nski gravity which satisfy a constraint of this theory.      
\end{abstract}

%***************************************************
\section{Introduction \label{intro}}
%***************************************************

In the late 70's of the previous century J. Kijowski \cite{kpt} introduced a method of constructing quantum states for field theories which was based on canonical quantization of local degrees of freedom and thereby was independent of the global structure of a spacetime. 

A classical field theory usually describes a physical system of an infinite number of (local) d.o.f.. Kijowski's method consists, roughly speaking, in doing three major steps:
\begin{enumerate}
\item first one isolates from the infinite system described by a field theory a set $\Lambda$ of finite physical systems (i.e. each such a system describes a finite number of d.o.f.);
\item next one ``quantizes'' each finite system $\lambda$ separately obtaining a space $\mathfrak{S}_{\lambda}$ of quantum states of the system;
\item finally one builds a space $\mathfrak{S}$ of quantum states for the theory from the family $\{\mathfrak{S}_\lambda\}_{\lambda\in\Lambda}$.
\end{enumerate}
The set $\La$ of finite systems chosen in the first step of the construction is required to be a directed set such that the directing relation $\geq$ is the relation {\em system-subsystem}, i.e. $\la'\geq\la$ if the system $\la$ is a subsystem of the system $\la'$. In the second step one assigns to the system $\la$ a Hilbert space $\h_\la$ and defines the space $\mathfrak{S}_\la$ of quantum states as the set of all states on a $C^*$-algebra $\B_\la$ of all bounded operators on $\h_\la$ \cite{mod-proj}. If the first step is  done properly then the family $\{\mathfrak{S}_\la\}_{\la\in\La}$ can be equipped with the structure of a {\em projective family} and the space $\mathfrak{S}$ is defined as the {\em projective limit} of the family. Therefore elements of $\mathfrak{S}$ are called {\em projective quantum states}.                 

The original projective method was applied by Kijowski to a theory of a scalar field \cite{kpt}. In recent years the method underwent gradual development related  generally speaking to attempts to quantize general relativity. This development allowed to apply the method to a background independent ``toy-theory'' called {\em degenerate Pleba\'nski gravity} (DPG) \cite{q-nonl}, to the Teleparallel Equivalent of General Relativity \cite{q-stat} and finally to Loop Quantum Gravity \cite{proj-lt-I,proj-lt-II,proj-lqg-I}. 

However, these constructions ignore dynamics of the theories---in all the cases listed above projective quantum states are built on the basis of phase spaces of the theories without taking into account constraints on the spaces and Hamiltonians. Therefore the states can be called {\em kinematic} or {\em unconstrained}. Needless to say that to be really useful in quantization of field theories  the projective method has to be supplemented by $(i)$ a procedure yielding {\em constrained} quantum states i.e. states which (in a sense) are solutions of constraints and are invariant with respect to gauge transformations generated by first class constraints and by $(ii)$ a procedure for implementing Hamiltonians. 

By now there are no detailed and complete procedures of these sorts which would be general enough to be applicable to a broad range of theories. There are only two partial ones: $(i)$ in \cite{kpt} there is an outline of a prescription for choosing a vacuum state from projective quantum states (this prescription refers to a Hamiltonian of a theory), and $(ii)$ in \cite{proj-lt-I} a rather general procedure was introduced aimed at transforming unconstrained finite physical systems of a theory into constrained ones i.e. into finite systems constructed from gauge invariant d.o.f. which satisfy constraints of the theory. It is worthy to note that the latter procedure was successfully applied in \cite{proj-lt-III} to a toy-model.

The goal of the present paper is to construct constrained projective quantum states for DPG and to check whether we can learn from this example anything which may be helpful in establishing a general procedure for constructing constrained projective quantum states. 

To be more precise: there are two constraints on the phase space of DPG. We will neglect one of them and will construct projective quantum states for the theory on the basis of classical solutions of the other constraint quotiented by corresponding gauge transformations. These constrained states will be built independently of the unconstrained ones constructed for DPG in \cite{q-nonl}. Comparing these two constructions we will be able to find a simple relation between unconstrained finite physical system and constrained ones. We will show then that this simple relation is not included in the procedure introduced in \cite{proj-lt-I} which means that this procedure should be suitably extended.   

Let us note finally that in \cite{non-comp} we constructed projective quantum states for DPG which solve both constraints of the theory, but that construction differs significantly from those presented in the papers listed above and rather cannot be generalized to be applicable to other theories.

The present paper is organized as follows: in Section \ref{prel} we describe briefly DPG and a detailed method for constructing unconstrained projective quantum states. The next section contains the construction of the constrained quantum states for the theory. In Section \ref{disc} we analyze the relation between the unconstrained and constrained finite systems and compare this relation to that assumed  by the procedure described in \cite{proj-lt-I}. Some proofs which were too technical to be included in the main part of the paper are placed at the end of it in Appendix.

%***************************************************
\section{Preliminaries \label{prel}}
%***************************************************

%***************************************************
\subsection{Degenerate Pleba\'nski gravity}
%***************************************************

Let ${\cal M}$ be a four-dimensional manifold, and let ${\cal G}=(\R,+)$ be the Lie group of the real numbers with the addition as the group action. Every connection $\mathbf{A}$ on a trivial principal bundle ${\cal M}\times {\cal G}$ can be naturally treated as a real-valued one-form on $\cal M$. Let $\bld{\sigma}$ and $\bld{\Theta}$ be, respectively, a two-form and a function on $\cal M$. The {\em degenerate Pleba\'nski gravity} (DPG) is a background independent theory of the fields $\bld{\sigma}$, $\mathbf{A}$ and $\bld{\Theta}$ subjected to Euler-Lagrange equations given by the following action:
\begin{equation}
\S[\bld{\sigma},\mathbf{A},\bld{\Theta}]:=\int_{\cal M} \Big(\bld{\sigma}\wedge d\mathbf{A} -\frac{1}{2}\bld{\Theta}\, \bld{\sigma}\wedge\bld{\sigma}\Big).
\label{action}
\end{equation} 
This theory was introduced in  \cite{non-comp} as a simplification of Pleba\'nski's self-dual formulation \cite{pleb} of general relativity and it describes a degenerate sector \cite{jac} of the latter theory.

To obtain a Hamiltonian formulation of DPG \cite{non-comp} it is convenient to suppose that there exists a three-dimensional compact connected manifold $\Sigma$ without boundary such that the product $\Sigma\times \R$ is diffeomorphic to ${\cal M}$---then one can regard the product as a $3+1$ decomposition of $\cal M$, where $\Sigma$ plays a role of a ``space'' and $\R$ plays a role of a 'time'. Canonical variables of DPG are a one-form\footnote{Here we use a version of canonical formalism in which a variable is not a component of a differential form but the entire form \cite{ham-diff}.} $A$ on $\Sigma$ as a configurational variable ($A$ can be thought of as a connection on a trivial principle bundle $\Sigma\times{\cal G}$) and a two-form $\sigma$ as a momentum conjugate to $A$. There are two first class constraints on the phase space: for every vector field $\vec{N}$ and for every function $\alpha$ on $\Sigma$   
\begin{align*}
C[\sigma,A;\vec{N}]&:=\int_{\Sigma}\sigma\we(\vec{N}\lr dA)=0, & C_{\rm G}[\sigma,A;\alpha]&:=\int_{\Sigma}\alpha \,d\sigma=0.
\end{align*}
The Hamiltonian of DPG is a sum of the two constraints
\[
H[\sigma,A;\alpha,\vec{N}]=C[\sigma,A;\vec{N}]+C_{\rm G}[\sigma,A;\alpha],
\]
where $\alpha,\vec{N}$ are Lagrange multipliers. 

There is no clear interpretation of gauge transformations generated by the constraint $C$, but if 
\[
C_{\rm D}[\sigma,A;\vec{N}]:=C[\sigma,A;\vec{N}]-C_{\rm G}[\sigma,A;\vec{N}\lr A]
\]
then orbits of gauge transformations generated by the constraint $C_{\rm D}$ are given by the following differential equations: 
\begin{align*}
\frac{d\sigma}{d\lambda}&=\{\sigma,C_{\rm D}[\sigma,A;\alpha]\}={\cal L}_{\vec{N}}\sigma,&\frac{dA}{d\lambda}&=\{A,C_{\rm D}[\sigma,A;\alpha]\}={\cal L}_{\vec{N}}A,
\end{align*}
where ${\cal L}_{\vec{N}}$ denotes the Lie derivative of a tensor field with respect to the vector field $\vec{N}$. Thus the gauge transformations defined by $C_{\rm D}$ are given by pull-backs of $\sigma$ and $A$ along integral curves of $\vec{N}$. Therefore $C_{\rm D}$ can be called a diffeomorphism constraint.     

In the case of the constraint $C_{\rm G}$ gauge transformations are solutions of the following equations: 
\begin{align*}
\frac{d\sigma}{d\lambda}&=\{\sigma,C_{\rm G}[\sigma,A;\alpha]\}=0,&\frac{dA}{d\lambda}&=\{A,C_{\rm G}[\sigma,A;\alpha]\}=-d\alpha.
\end{align*}
This means that the momentum $\sigma$ is preserved by the transformations, while the one-form $A$ is transformed by adding to it the differential of a function $f$ on $\Sigma$:
\[
A\mapsto A'=A+df.
\]  
Thus $C_{\rm G}$ is a {\em Gauss constraint}.

Below we will neglect the diffeomorphism constraint and will construct constrained projective quantum states for DPG using solutions of the Gauss constraint and gauge transformations defined by it. Since this theory is background independent the unconstrained states described in \cite{q-nonl} were constructed in a background independent way. The constrained states will be built in the same manner.

%***************************************************
\subsection{Projective quantum states---outline of the construction \label{outline}}
%***************************************************

The idea of the construction of projective quantum states for field theories, as described in Section \ref{intro}, is rather simple, but in practice the first step of the construction, that is, a proper choice of a set $\La$ of finite physical systems may be very difficult. In this paper to choose properly a set $\La$ which will give us constrained states for DPG  we will apply a prescription presented in \cite{q-nonl}. This prescription yields unconstrained states and we will have to modify it to obtain the desired constrained ones.

The prescription  assumes that a phase space of a field theory is a product $P\times Q$ of a space $P$ of momenta and a configuration space $Q$ and requires to choose $(i)$ a set $\cal K$ of real functions on $Q$ which separates points in $Q$ and $(ii)$ a set $\cal F$ of real functions on $P$ which separates points in $P$. Elements of $\cal K$ ($\cal F$) are called {\em configurational (momentum) elementary degrees of freedom}.          
 
Consider a finite set $K=\{\kappa_1,\ldots,\kappa_N\}$ of configurational elementary d.o.f. and denote by $[q]$ an intersection of level sets of the  functions $\{\kappa_1,\ldots,\kappa_N\}$ which contains $q\in Q$. Let
\[
Q_K:=\{\ [q] \ | \ q\in Q\ \}.
\]
The functions $\{\kappa_1,\ldots,\kappa_N\}$ define a natural map\footnote{In fact, there are many maps of this sort which differ from each other by the ordering of the numbers $\big(x_1([q]),\ldots,x_N([q])\big)$. Since the construction of projective quantum states does not prefer any particular ordering we will regard every such a map as natural.}
\begin{equation}
\begin{aligned}
  \tilde{K}:&\, Q_K\to \R^N,\\
  &[q]\mapsto \big(x_1([q]),\ldots,x_N([q])\big):=\big(\kappa_1(q),\ldots,\kappa_N(q)\big),
\end{aligned}  
\label{K-fr}
\end{equation}
which will be treated as a (global) coordinate frame on $Q_K$.

The d.o.f. $\{\kappa_1,\ldots,\kappa_N\}$ are said to be {\em independent} if the image of the map \eqref{K-fr} is an open subset of $\R^N$. Then $Q_K$ can be treated as a smooth manifold called  {\em reduced configuration space}.      

Let $K$ be a set of independent d.o.f.. A {\em cylindrical function compatible with} $K$ is a complex function $\Psi$ on $Q$ such that 
\begin{equation}
\Psi=\pr_K^*\psi,
\label{cyl-f}
\end{equation}
where $\pr_K$ is a projection 
\[
Q\in q\mapsto \pr_K(q):=[q]\in Q_K
\]  
and $\psi$ is a smooth complex function on $Q_K$. A complex linear space spanned by all  cylindrical functions (given by all sets $\{K\}$ of independent d.o.f.) will be denoted by $\Cyl$. 

Assume that every momentum d.o.f. $\varphi\in{\cal F}$ and the Poisson bracket  on $P\times Q$ define a linear operator $\hat{\varphi}$ on $\Cyl$:
\begin{equation}
\Cyl\ni\Psi\mapsto\hat{\varphi}\Psi:=\{\varphi,\Psi\}\in\Cyl
\label{mom-op}
\end{equation}
which will be called a {\em momentum operator}. In some cases including the one considered in this paper it is necessary to admit operators $\hat{\varphi}$ defined on $\Cyl$ via a regularization of $\{\varphi,\Psi\}$. Denote by $\hat{\cal F}$ a real linear space spanned by all momentum operators $\{\hat{\varphi}\}$.   

Let $\mathbf{K}$ be the set of all sets of independent configurational elementary d.o.f. and let $\hat{\mathbf{F}}$ be the set of all finite dimensional linear subspaces of $\hat{\F}$. Suppose that $\Lambda$ is a subset of a product $\hat{\mathbf{F}}\times \mathbf{K}$. Then an element $\la$ of $\La$ is a pair $(\hat{F},K)$, where the finite dimensional space $\hat{F}$ represents a finite number of momentum elementary d.o.f. and the set $K$ represents a finite number of configurational elementary d.o.f.. Thus $\la$ can be regarded as a finite physical system.

Suppose moreover that $\La$ is equipped with a directing relation $\geq$ such that the directed set $(\La,\geq)$ satisfies the following {\em Assumptions}:
\begin{enumerate}
\item 
\begin{enumerate}
\item for every finite set $K_0$ of configurational elementary d.o.f. there exists $(\hat{F},K)$ $\in\Lambda$ such that every $\kappa\in K_0$ is a cylindrical function compatible with $K$; \label{k-Lambda}
\item for every finite set $F_0$ of momentum elementary d.o.f. there exists $(\hat{F},K)\in\Lambda$ such that for every $\varphi\in F_0$ the corresponding operator $\hat{\varphi}\in\hat{F}$; \label{f-Lambda}
\end{enumerate}
\item \label{RN} 
if $(\hat{F},K)\in\Lambda$ then the image of the map \eqref{K-fr} is $\R^N$, where $N$ is the number of elements of $K$;
\item 
if $(\hat{F},K)\in\Lambda$, then 
\begin{enumerate}
\item for each $\hat{\varphi}\in \hat{\F}$ and for each cylindrical function $\Psi=\pr_K^*\psi$ compatible with $K=\{\kappa_1,\ldots,\kappa_N\}$ 
\[
\hat{\varphi}\Psi=\sum_{I=1}^N\Big(\pr^*_K\partial_{x_I}\psi\Big)\hat{\varphi}\kappa_I,
\]   
where $\{\partial_{x_i}\}$ are vector fields on $Q_K$ defined by the coordinate frame \eqref{K-fr}; \label{comp-f} 
\item for each $\hat{\varphi}\in \hat{\F}$ and for each $\kappa\in K$ the cylindrical function $\hat{\varphi}\kappa$ is a real {\em constant} function on $Q$; \label{const}
\end{enumerate}
\item if $(\hat{F},K)\in\Lambda$ and $K=\{\kappa_{1},\ldots,\kappa_{N}\}$ then $\dim\hat{F}=N$ and if $(\hat{\varphi}_1,\ldots,\hat{\varphi}_N)$ is a basis of $\hat{F}$ then an $N\times N$ matrix $G=(G_{JI})$ of components
\begin{equation}
G_{JI}:=\hat{\varphi}_J\kappa_I
\label{Gij}
\end{equation}
is {\em non-degenerate}; \label{non-deg}
\item  if $(\hat{F},K'),(\hat{F},K)\in\Lambda$ and $Q_{K'}=Q_{K}$ then  $(\hat{F},K')\geq(\hat{F},K)$; \label{Q'=Q} 
\item if $(\hat{F}',K')\geq(\hat{F},K)$ then 
\begin{enumerate}
\item every d.o.f. in $K$ is {\em a linear combination} of d.o.f. in $K'$; \label{lin-comb}
\item every operator in $\hat{F}$ is {\em a linear combination} of operators in $\hat{F}'$. \label{FF'}
\end{enumerate} 
\end{enumerate} 
If these assumptions are satisfied and if for every $\la=(\hat{F},K)\in\La$ 
\[
\h_\la:=L^2({Q}_{K},d\mu_\la),
\]
where $d\mu_\la$ is a Lebesgue measure on ${Q}_{K}$ defined by the natural coordinate frame \eqref{K-fr} then the resulting family $\{\mathfrak{S}_\la\}_{\la\in\La}$  is naturally a projective family and its projective limit $\mathfrak{S}$ is ``large enough'' to serve as a space of quantum states for the field theory \cite{mod-proj}.  

%***************************************************
\section{Constrained projective quantum states for DPG}
%***************************************************

Recall that the phase space of DPG is a product $P\times Q$, where the space $P$ of momenta is the set of all two-forms on $\Sigma$, and the configuration space $Q$ is the set of all one-forms on the manifold. A point $(\sigma,A)$ of the phase space satisfies the Gauss constraint $C_{\rm G}$ if $d\sigma=0$. Consequently, if $P_{\rm G}$ is the set of all closed two-forms on $\Sigma$ then $P_{\rm G}\times Q$ is the set of all (classical) solutions of the constraint. Let us now define on $Q$ an equivalence relation: $A,A'\in Q$ are said to be equivalent, $A\sim A'$, if there exists a function $f$ on $\Sigma$ such that $A'=A+df$. Denote by $\tilde{Q}$ a quotient space $Q/\!\!\sim$. Then the product $P_{\rm G}\times \tilde{Q}$ is the set of all orbits of the gauge transformations generated by the Gauss constraint contained in  $P_{\rm G}\times Q$ (recall that the transformations preserve each momentum $\sigma\in P_{\rm G}$). In other words $P_{\rm G}\times \tl{Q}$ is the set of all (classical) solutions of the constraint modulo the gauge transformations. 

To construct the desired constrained projective  quantum states for DPG  we will apply a modification of the prescription described in the previous subsection---the modification will be simple and natural and will consist in choosing the space $P_{\rm G}\times\tilde{Q}$ as the point of departure for the construction instead of the phase space $P\times Q$. To put it differently, we solved the Gauss constraint and will ``quantize'' the orbits of the gauge transformations passing through the solutions obtaining thereby the quantum states. Let us emphasize that in the sequel the letter $Q$ appearing in symbols used in the prescription will be suitably replaced by $\tilde{Q}$.

According to the prescription the first step of the construction of projective quantum states is a choice of elementary d.o.f.. Because we would like our construction to be background independent we will apply some ideas used in Loop Quantum Gravity (LQG) (see e.g. \cite{cq-diff,rev,rev-1} and references therein) to define the d.o.f.: since our configurational variable is a connection one-form $A$ a configuration  elementary d.o.f. will be defined as a holonomy of the connection along a loop in $\Sigma$, similarly, since our momentum variable is a two-form a momentum elementary d.o.f. will be defined as an integral of it over a surface in $\Sigma$. 

To allow the LQG ideas to work we suppose  that the manifold $\Sigma$ is {\em real analytic} and {\em oriented}. 

%***************************************************
\subsection{Configurational elementary d.o.f.}
%***************************************************

%***************************************************
\subsubsection{Loops} 
%***************************************************

An {\em analytic edge} is a one-dimensional connected analytic embedded submanifold of $\Sigma$ with a two-point boundary i.e. with a boundary which consists of two {\em distinct} points. An {\em oriented} one-dimensional connected $C^0$ submanifold of $\Sigma$ with two-point boundary given by a finite union of analytic edges will be called an {\em edge}. Due to the orientation we can distinguish two points constituting the boundary of an edge: one of them can be called a {\em source} and the other a {\em target} of the edge. An edge $e^{-1}$ is an {\em inverse} of an edge $e$ if $e^{-1}$ and $e$ coincide as unoriented submanifolds of $\Sigma$ and possess opposite orientations.  

Let $(e_1,e_2,\ldots,e_{N-1},e_N)$ be an ordered set of edges such that the target of the edge $e_{I}$ coincides with the source of the edge $e_{I+1}$. The set will be called a {\em composition} of the edges and denoted by 
\[
e_N\circ e_{N-1}\circ \ldots\circ e_{2}\circ e_1.
\]
Note that a composition of edges as defined here is never an edge but in some cases it can be identified with an edge. Suppose that $(y_1,\ldots,y_N)$ is a sequence of pairwise distinct points of the interior of an edge $e$ ordered according to its orientation. The points divide the edge into subedges $\{e_1,\ldots,e_N,e_{N+1}\}$ which can be composed as $e_{N+1}\circ e_N\circ\ldots\circ e_1$. Obviously, in this case it is natural to identify the composition with the original edge $e$.      

A {\em (piecewise analytic) loop} $l$ based at $y\in\Sigma$ is a composition $e_N\circ\ldots\circ e_1$ of edges such that both the source of $e_1$ and the target of $e_N$ coincide with $y$. The {\em inverse} $l^{-1}$ of the loop $l$ is a loop based at the same point as $l$ given by a composition  $e^{-1}_1\circ\ldots\circ e^{-1}_N$. If loops $l=e_N\circ\ldots\circ e_1$ and $l'=e'_N\circ\ldots\circ e'_1$ are based at the same point then a composition $l'\circ l$ is a loop defined as 
\[
e'_N\circ\ldots\circ e'_1\circ e_N\circ\ldots\circ e_1.
\]
The set of all loops in $\Sigma$ will be denoted by $\mathbb{L}$.

%***************************************************
\subsubsection{Definition of configurational elementary d.o.f.}
%***************************************************

An integral of a one-form $A$ along a composition  $e_N\circ\ldots\circ e_1$ of edges can be defined naturally as follows:
\[
\int_{e_N\circ\ldots\circ e_1} A:=\sum_{I=1}^N \int_{e_I}A.
\]
Denote by $[A]_{\rm G}$ an orbit of the gauge transformations generated by the Gauss constraint $C_{\rm G}$ which contains the one-form $A$. Let $l$ be a loop. Then 
\[
\tilde{Q}\ni[A]_{\rm G}\mapsto \kappa_l([A]_{\rm G}):=\int_l A \in\R
\]  
is a well defined map---if $A'\in[A]_{\rm G}$ then $A'=A+df$ and
\[
\int_lA'=\int_lA+\int_ldf=\int_lA.
\]  
To keep the notation as simple as possible we will write $\kappa_l(A)$ instead of $\kappa_l([A]_{\rm G})$. As shown in \cite{oko-H} $\kappa_l(A)$ is nothing else but a holonomy of the connection $A$ along the loop $l^{-1}$.   

We choose a set
\begin{equation}
\{\ \kappa_l \neq 0\ | \ l\in\mathbb{L}  \ \}
\label{cedof}
\end{equation}
to be a set $\K$ of configurational elementary d.o.f..

A loop $l$ will be called {\em trivial} if $\kappa_l=0$.

Suppose that the set $\K$ does not separate points in $\tilde{Q}$ i.e. that there exists one-forms $A,A'\in Q$ such that $[A]_{\rm G}\neq [A']_{\rm G}$ and for every $l\in\mathbb{L}$ $\kappa_l(A)=\kappa_l(A')$. But then $\kappa_{l}(A-A')=0$ for each loop $l$. On the other hand  if an integral of a one-form along every (piecewise analytic) loop is zero then the one-form is exact. Thus $A-A'$ is exact which means that $[A]_{\rm G}= [A']_{\rm G}$ contrary to the assumption above. This allows us to conclude that the set $\K$ does separate points in $\tilde{Q}$.

%***************************************************
\subsubsection{Finite sets of configurational d.o.f.}
%***************************************************

Here we are going to distinguish a sort of finite sets of configurational d.o.f. which will be used then to define a directed set of finite physical systems.

To this end we will change a way of labeling the configurational d.o.f.. A labeling we are going to use was introduced in \cite{al-hoop}.  It is based on an observation that many distinct loops define the same d.o.f.. For example, if $l=e_2\circ e_1$ and $l'=e_2\circ e^{-1}\circ e\circ e_1$ then obviously $\kappa_l=\kappa_{l'}$. Therefore it is reasonable to identify those loops which define the same d.o.f.: we say that a loop $l$ is in relation with a loop $l'$, $l\sim l'$, if for every $A\in Q$     
\[
\kappa_l(A)=\kappa_{l'}(A).
\]
The relation $\sim$ is an equivalence one. An equivalence class of a loop $l$ will be denoted by $\tilde{l}$ and called a {\em hoop}. Since now we will also use a symbol $\kappa_{\tilde{l}}\equiv \kappa_l$. The set of all hoops will be denoted by $\h\G$.

Let $y$ be an arbitrary point of $\Sigma$. Then each hoop $\tilde{l}$ contains a loop based at $y$. Indeed, let the loop $l$ be based at $y'$. Since $\Sigma$ is assumed to be connected there exists an edge $e$ which starts at $y$ and ends at $y'$. Then $l':=e^{-1}\circ l\circ e$ is based at $y$ and belongs to $\tilde{l}$. 

This fact allows to introduce a composition of hoops $\tilde{l}$ and $\tilde{l}'$ \cite{al-hoop}. Let loops $l_y$ and $l'_y$ be based at $y$ and be representatives of the hoops. Then 
\[
\tilde{l}\circ\tilde{l}':=\widetilde{l_y\circ l'_y}.
\]                 
One can easily show that this composition does not depend on the choice of $y$ and the representatives of the hoops. Moreover, it defines on $\h\G$ a structure of an {\em Abelian group} (the neutral element of the group is a hoop given by a trivial loop, the inverse $\tilde{l}^{-1}$ of $\tilde{l}$ is equal to $\widetilde{l^{-1}}$).  

Consider a set $\{l_1,\ldots,l_N\}$ of non-trivial loops such that
\begin{equation}
l_I=e_{In_I}\circ\ldots \circ e_{I1}. 
\label{li-ee}
\end{equation}
We say that the loops are {\em independent} if \cite{al-hoop}
\begin{enumerate}
\item for each $I$ and for each $i<n_I$ there exists an open neighborhood $U$ of the target of $e_{Ii}$ such that $e_{Ii}\cap e_{I(i+1)}\cap U$ coincides with the target\footnote{This condition excludes a case $e_{I(i+1)}=e^{-1}_{Ii}$ and similar ones.};
\item each loop $l_I$ contains an edge $e_{Im_I}$ (being a factor in the composition \eqref{li-ee}) such that 
\begin{enumerate}
\item for every $i\neq m_I$ an intersection $e_{Im_I}\cap e_{Ii}$  is a finite or the empty set;  
\item for every $J\neq I$ an intersection\footnote{An intersection of a set $V\subset \Sigma$ and a loop $l=e_N\circ\ldots \circ e_1$ is naturally defined as $V\cap (e_N\cup\ldots \cup e_1)$. An intersection of the loop $l$ and a loop $l'=e'_{N'}\circ\ldots \circ e'_1$ is a set $(e_N\cup\ldots \cup e_1)\cap (e'_{N'}\cup\ldots \cup e'_1)$.} $e_{Im_I}\cap l_J$ is a finite or the empty set. 
\end{enumerate}
\end{enumerate}
Non-trivial hoops $\{\tilde{l}_1,\ldots,\tilde{l}_N\}$ are {\em independent} if there exists a representative $l'_I$ of each $\tilde{l}_I$ such that the loops $\{l'_1,\ldots,l'_N\}$ are independent.   

\begin{lm}{\rm \cite{al-hoop}}
For every finite set $\{\tilde{l}_1,\ldots,\tilde{l}_{N}\}$ of hoops there exists a set $\tilde{L}'$ of independent hoops such that each hoop $\tilde{l}_I$ is a composition of hoops in $\tilde{L}'$ and their inverses. \label{h-h}
\end{lm}
\noindent This lemma implies that the set $\tilde{\mathbf{L}}$ of all sets of independent hoops is a directed set with a directing relation defined as follows: $\tilde{L}'\geq \tilde{L}$ if each hoop in $\tilde{L}$ is a composition of hoops in $\tilde{L}'$ and their inverses \cite{proj}. 

Each set $\tilde{L}=\{\tilde{l}_1,\ldots,\tilde{l}_N\}$ of independent hoops defines a finite set of configurational elementary d.o.f.:
\[
K_{\tilde{L}}:=\{\kappa_{\tilde{l}_1},\ldots,\kappa_{\tilde{l}_N}\}.
\]

\begin{lm}{\rm \cite{al-hoop}}
Let $\tilde{L}=\{\tilde{l}_1,\ldots,\tilde{l}_N\}$ be a set of independent hoops. Then for every $(x_1,\ldots,x_N)\in \R^N$ there exists a one-form $A\in Q$ such that
\[
\kappa_{\tilde{l}_I}(A)=x_I.
\]
\label{K-bij}
\end{lm}
\noindent It follows from the lemma that the image of the map $\tilde{K}_{\tilde{L}}$ (see the formula \eqref{K-fr}) is $\R^N$ (with $N$ being the number of elements of $\tilde{L}$) which means that the d.o.f. in $K_{\tilde{L}}$ are independent and the space $\tilde{Q}_{K_{\tilde{L}}}$ is a reduced configuration space on which the map $\tilde{K}_{\tilde{L}}$ defines a linear structure by a pull-back of the linear structure on $\R^N$.     

It is obvious that for any hoops $\tilde{l},\tilde{l}_1,\tilde{l}_2$ 
\begin{align}
\kappa_{\tilde{l}^{-1}}&=-\kappa_{\tilde{l}}, & \kappa_{\tilde{l}_1\circ \tilde{l}_2}=&\kappa_{\tilde{l}_1}+\kappa_{\tilde{l}_2}.
\label{kkk}
\end{align}
These equations and Lemma \ref{h-h} give us
\begin{lm}
  For every finite set $K_0$ of configurational elementary d.o.f. there exists a set $\tilde{L}$ of independent hoops such that every d.o.f. in $K_0$ is a linear combination of d.o.f. in $K_{\tilde{L}}$.
  \label{K0-lin-KL}
\end{lm}
\noindent It turns out that 
\begin{lm}
Let $\tl{L}$ be a set of independent hoops. If a d.o.f. $\kappa_{\tilde{l}}$ is a linear combination of d.o.f. in $K_{\tilde{L}}$ then the hoop $\tilde{l}$ is a composition of hoops in $\tilde{L}$ and their inverses.
\label{k-akak}
\end{lm}
\noindent For a proof of this lemma see Appendix \ref{proof}.
\begin{pro}
Let $\tilde{L}',\tilde{L}$ be sets of independent hoops. Then $\tilde{L}'\geq\tilde{L}$ if and only if each d.o.f. in $K_{\tilde{L}}$ is a linear combination of d.o.f. in $K_{\tilde{L}'}$.   
\label{L'L-lin}
\end{pro}

\begin{proof}
If $\tilde{L}'\geq\tilde{L}$ then by virtue of Equations \eqref{kkk} each d.o.f. in $K_{\tilde{L}}$ is a linear combination of d.o.f. in $K_{\tilde{L}'}$.

If every d.o.f. in $K_{\tilde{L}}$ is a linear combination of d.o.f. in $K_{\tilde{L}'}$ then by virtue of Lemma \ref{k-akak} every hoop in $\tilde{L}$ is a composition of hoops in  $\tilde{L}'$ and their inverses. Therefore $\tilde{L}'\geq\tilde{L}$. 
\end{proof}

%***************************************************
\subsubsection{Cylindrical functions}
%***************************************************

It follows from Lemma \ref{K-bij} that there are finite sets of configurational elementary d.o.f. \eqref{cedof} which define reduced configuration spaces. This fact allows us to consider cylindrical functions on $\tilde{Q}$ given by smooth functions on these spaces. Now let us check whether a linear space $\Cyl$ spanned by all the cylindrical functions will serve its purpose which is to be a domain of momentum operators \eqref{mom-op} associated with momentum elementary d.o.f. To this end we will apply 

\begin{pro}
{\rm \cite{q-nonl} } Let $\mathbf{K}$ be the set of all sets of independent configurational d.o.f.. Suppose that there exists a subset $\mathbf{K}'$ of $\mathbf{K}$ such that  for every finite set $K_0$ of configurational elementary d.o.f. there exists $K'_0\in\mathbf{K}'$ satisfying the following conditions: 
\begin{enumerate}
\item the map $\tilde{K}'_0$ is a bijection; 
\item each d.o.f. in $K_0$ is a linear combination of d.o.f. in $K'_0$.  
\end{enumerate} 
Then
\begin{enumerate}
\item if $Q_{K}=Q_{K'}$ for $K,K'\in\mathbf{K}$ then the differential structures defined on the space by $\tilde{K}$ and $\tilde{K}'$ coincide.
\item for every element $\Psi\in\Cyl$ there exists a set $K\in\mathbf{K}'$ such that $\Psi$ is compatible with $K$.    
\end{enumerate}
\label{big-pro}
\end{pro}
\noindent Note that due to Lemmas \ref{K-bij} and \ref{K0-lin-KL} the set $\{K_{\tilde{L}}\}_{\tilde{L}\in\tilde{\mathbf{L}}}$ satisfies assumptions imposed in the proposition on the set $\mathbf{K}'$. Thus the assertions of the propositions are valid in the case of our construction. 

The first assertion of the proposition guarantees that in our construction every reduced configurational space $\tilde{Q}_K$ possesses a unique differential structure even if there exists a set $K'\neq K$ of independent d.o.f. such that $\tl{Q}_K=\tl{Q}_{K'}$. Note that if it was not the case then distinct differential structures on the same reduced configuration space would  be a source of problems taking into account that the action of momentum operators on cylindrical functions involves a differentiation of functions defined on reduced configuration spaces (see \eqref{cyl-f} and \eqref{mom-op}).   

The second assertion will allow us to define momentum operators by describing their actions on cylindrical functions compatible with sets of d.o.f. given by sets of independent hoops.

%***************************************************
\subsection{Momentum elementary d.o.f.}
%***************************************************

%***************************************************
\subsubsection{Faces \label{sec-lf}}
%***************************************************
Let $S$ be an analytic two-dimensional embedded submanifold of $\Sigma$. We say that an edge $e$ is adapted to $S$ if either \cite{area} 
\begin{enumerate}
\item $e$ is contained in the closure $cl(S)$ or has no common points with $S$;  
\item $e$ has exactly one common point with $S$ being either the source or the target of $e$.    
\end{enumerate}
In the latter case the edge $e$ will be called {\em transversal to} $S$.

Consider now a loop 
\begin{equation}
l=e_n\circ\ldots\circ e_1
\label{l-ad-ee}
\end{equation}
composed from edges adapted to the submanifold $S$ and assume that $S$ is {\em oriented}. We will now associate with $S$ and $l$ a half-integer $\eps(S,l)$ which will be used later on to define momentum operators.      
   
To define the half-integer $\eps(S,l)$ let us count those edges in \eqref{l-ad-ee} which are transversal to $S$ treating them as factors of the composition i.e. if two edges $e_i$ and $e_j$ appear in the composition at different positions then they will be treated as distinct even if they coincide as oriented submanifolds of $\Sigma$. Let      
\begin{enumerate}
\item $t^+$  be the number of factors in \eqref{l-ad-ee} which intersect $S$ only at their targets and are placed 'above' the face;
\item $s^+$ be the number of factors in \eqref{l-ad-ee} which intersect $S$ only at their sources and are placed 'above' the face;
\item $t^-$  be the number of factors in \eqref{l-ad-ee} which intersect $S$ only at their targets and are placed 'below' the face;
\item $s^-$ be the number of factors in \eqref{l-ad-ee} which intersect $S$ only at their sources and are placed 'below' the face.
\end{enumerate}
Here the terms 'below' and 'above' refer to the orientation of the normal bundle of $S$ which is defined by the orientations of $S$ and $\Sigma$. Then 
\begin{equation}
\eps(S,l):=\frac{1}{2}\big(t^+-s^+-(t^--s^-)\big).
\label{eps-0}
\end{equation}

An analytic two-dimensional embedded submanifold $S$ of $\Sigma$ will be called a {\em face} if
\begin{enumerate}
\item $S$ is connected, oriented and its closure $cl(S)$ is compact;
\item every edge $e$ in $\Sigma$ is a composition $e_N\circ\ldots\circ e_1$ of edges {\em adapted} to $S$;
\item there exists a non-trivial loop $l$ composed of edges adapted to $S$ such that $\eps(S,l)\neq 0$. 
\end{enumerate}

Let us now comment on the last requirement in the definition of a face since to the best of our knowledge it has not been introduced before (the remaining requirements imposed on a face can be found in e.g. \cite{area}, \cite{acz} and \cite{q-nonl}). This requirement may seem to be formulated in a quite complicated way, but the only reason we introduced it is that it will guarantee immediately that for every face  a momentum operator associated with the face is non-zero. Therefore we do not need any simpler formulation of the requirement.

If a submanifold $S$ satisfies all the requirements of the definition of a face except the last one and if $cl(S)\setminus S$ is non-empty then $S$ meets the last requirement as well and thereby is a face. To see this let us fix a point $y$ in $cl(S)\setminus S$. Then we can find a point $y'\in S$ such that $y$ and $y'$ can be connected  by $(i)$ an edge $e_1$ contained in $cl(S)$ and by $(ii)$ an edge $e_2$ such that $e_2\cap S=y'$. Once suitably oriented the edges $e_1$ and $e_2$ compose a loop $l$ such that $\eps(S,l)=1/2$.            

If $cl(S)\setminus S$ is empty then $S$ may be or may be not a face. For example, a two-sphere $S^2$ once embedded in $\Sigma$ in such a way that it is a boundary of a three-dimensional ball in $\Sigma$  does not meet the last requirement of the definition. But if $\Sigma=S^2\times S^1$, where $S^1$ is a circle and if $y$ is a point of $S^1$ then $S=S^2\times\{y\}$ is a face\footnote{I am very grateful to Wojciech Kami\'nski for suggesting me this example.}. Indeed, if $y'$ is a point in $S^2$ and a loop $l$ is chosen to be $\{y'\}\times S^1$ (it is a trivial task to express this loop as a composition of edges adapted to $S$) then $\eps(S,l)=\pm 1$ where the sign depends on the orientations of $S$ and $l$.                  

The set of all faces in $\Sigma$ will be denoted by $\mathbb{S}$.

%***************************************************
\subsubsection{Definition of momentum elementary d.o.f.}
%***************************************************

Let $S$ be a face. It defines a function on the space $P_{\rm G}$ of those momenta which satisfy the Gauss constraint: 
\[
P_{\rm G}\ni\sigma\mapsto \varphi_S(\sigma):=\int_S\sigma\in\R.
\]

We choose a set
\[
\{\ \varphi_S\neq 0 \ | \ S\in\mathbb{S}\ \}
\]
to be a set $\F$ of momentum elementary d.o.f.. It is not difficult to check that this set separates points in $P_{\rm G}$.

%***************************************************
\subsubsection{Momentum operators}
%***************************************************

It turns out that for any $\Psi\in\Cyl$ and $\varphi_S\in\F$ the Poisson bracket $\{\varphi_S,\Psi\}$ is ill defined, but by means of  a regularization \cite{acz} it is possible to obtain from it a linear operator $\hat{\varphi}_S$ on $\Cyl$ (called  a {\em flux operator} in the LQG literature). The action of the operator on elements of $\Cyl$ reads  
\begin{equation}
\hat{\varphi}_S\Psi:=\sum_{I=1}^N\eps(S,\tilde{l}_I)\, \pr_{K_{\tilde{L}}}^*(\partial_{x_I}\psi),
\label{hphi_S}
\end{equation}
where $(i)$ $\Psi\in\Cyl$ is treated as a cylindrical function $\pr_{K_{\tilde{L}}}^*\psi$ compatible with a set $K_{\tilde{L}}$ given by a set $\tilde{L}=\{\tll_1,\ldots,\tll_N\}$ of independent hoops (as it is allowed by Proposition \ref{big-pro}), $(ii)$ $\{\partial_{x_1},\ldots,\partial_{x_N}\}$ are vector fields on $\tilde{Q}_{K_{\tilde{L}}}$ defined by the natural coordinate frame $(x_1,\ldots,x_N)$ on the space (see \eqref{K-fr}) and $(iii)$ $\eps(S,\tilde{l}_I)$ is a real number associated with the hoop $\tilde{l}_I$ according to the following prescription.      

Consider a loop $l_0\in\tilde{l}_I$. Each edge composing the loop can be composed from edges adapted to the face $S$ which means that there is a loop 
\begin{equation}
l=e_n\circ\ldots\circ e_1
\label{l-ad-ee-1}
\end{equation}
in $\tilde{l}_I$ such that each edge $e_i$ is adapted to $S$. Then
\begin{equation}
\eps(S,\tilde{l}_I):=\eps(S,l)
\label{eps},
\end{equation}
where $\eps(S,l)$ is given by \eqref{eps-0}.
 
The above definition of the momentum operator $\hat{\varphi}_S$ requires to choose $(i)$ a set of independent hoops $\tilde{L}$ to describe the function $\Psi$ as a pull-back of a function $\psi$  and $(ii)$ a loop in each hoop $\tilde{l}_I\in \tilde{L}$. Therefore we have to show that the resulting operator do not depend on these choices. Since a proof of this fact is quite technical we relegated it to Appendix \ref{hphi-ok}.   

Taking into account Equations \eqref{hphi_S} and \eqref{eps} it is clear that due to the last requirement in the definition of a face every operator $\hat{\varphi}_S$ is non-zero.  

Assume that a hoop $\tll$ is a composition of independent hoops $\{\tll_1,\ldots,\tll_N\}=\tl{L}$ and their inverses (such a set $\tl{L}$ exists for every hoop due to Lemma \ref{h-h}). Then by virtue of \eqref{kkk}
\[
\kappa_{\tll}=\sum_{I=1}^Nn_I\kappa_{\tll_I},
\]
where $\{n_I\}$ are integers. This means that $\kappa_{\tll}$ is a cylindrical function compatible with $K_{\tl{L}}$. Then for every face $S$    
\begin{equation}
\hat{\varphi}_S\kappa_{\tll}=\sum_{I=1}^Nn_I\eps(S,\tll_I).
\label{hphiS-kl}
\end{equation}
Thus for every hoop $\tll$ and every face $S$   $\hat{\varphi}_S\kappa_{\tll}$ is a real {\em constant} cylindrical function. 

Let $\hat{\F}$ be a real linear space spanned by all momentum operators $\{ \hat{\varphi}_S \}_{S\in\mathbb{S}}$. Any element $\hat{\varphi}\in \hat{\F}$ is of the following form
\[
\hat{\varphi}=\sum_i \alpha_i\hat{\varphi}_{S_i},
\]  
where $\{\alpha_i\}$ are real numbers and the sum is finite. Let $\tl{L}=\{\tll_1,\ldots,\tll_N\}$ be a set of independent hoops and  $\Psi=\pr_{K_{\tl{L}}}\psi$ a cylindrical function compatible with ${K_{\tl{L}}}$. By virtue of  Equations \eqref{hphi_S} and \eqref{hphiS-kl} the action of the operator $\hat{\varphi}$ on $\Psi$ reads:
\begin{equation}
\hat{\varphi}\Psi=\sum_{iI}\alpha_i\eps(S_i,\tll_I)\, \big(\pr_{K_{\tl{L}}}^*\partial_{x_I}\psi\big)=\sum_i\big(\pr_{K_{\tl{L}}}^*\partial_{x_I}\psi\big)\,\hat{\varphi}\kappa_{\tll_I}.
\label{hphi-Psi}
\end{equation}       

%***************************************************
\subsection{A comment on analyticity of edges and faces}
%***************************************************

Let us finally explain why we assumed edges and faces to be (piecewise) analytic\footnote{Analyticity can be weakened to a property called {\em semi-analyticity} without spoiling the desired effect \cite{fl,lost}.}. The reason for this is that intersections of (piecewise) analytic submanifolds are of rather simple sort: $(i)$ the intersection of two piecewise analytic edges consists of a finite set of edges and/or a finite number of points \cite{al-hoop} and $(ii)$ the intersection of a piecewise analytic edge and an analytic surface consists of a set of edges and/or a finite number of points \cite{acz} (here the set of edges may be infinite if e.g. $S$ is obtained by removing from an analytic disc an infinite sequence of separate points converging to a point outside the disc; therefore defining a face we excluded this possibility requiring that every edge can be represented by a  composition of a {\em finite} number of edges adapted to the face---see Section \ref{sec-lf}.) In other words, each such an intersection does not contain an infinite number of isolated points which is essential for some constructions which were taken from LQG like the directed set $\tl{\mathbf{L}}$ of sets of independent hoops and the momentum operators \eqref{hphi_S}.

%***************************************************
\subsection{Directed set of finite physical systems \label{dir-set}}
%***************************************************

In the previous sections we have defined configurational elementary d.o.f. and momentum operators and described some of their properties. Now we are ready to organize them into a directed set $(\La,\geq)$ of (constrained) finite physical systems which will provide us with the desired constrained projective quantum states for DPG. 

To obtain the set $(\La,\geq)$ let us use a general construction described in \cite{lqg-tens}. Recall that in Section \ref{outline}  we denoted by $\mathbf{K}$ the set of all sets of independent d.o.f.  and by $\hat{\mathbf{F}}$ the set of all finite dimensional spaces of momentum operators chosen for a theory. A pair $(\hat{F},K)\in\hat{\mathbf{F}}\times \mathbf{K}$, where $K=\{\kappa_1,\ldots,\kappa_N\}$, is said to be {\em non-degenerate} if 
\begin{enumerate}
\item for every $\hat{\varphi}\in \hat{F}$ and every $\kappa_I\in K$ the function $\hat{\varphi}\kappa_I$ is a real {\em constant} function, 
\item $\dim \hat{F}=N$,
\item for a basis $(\hat{\varphi}_1,\ldots,\hat{\varphi}_N)$ of $\hat{F}$ the matrix \eqref{Gij} is non-degenerate.
\end{enumerate}

Suppose that $(Z,\geq)$ is a directed set and that elements of it label elements of $\mathbf{K}_Z \subset \mathbf{K}$, i.e.:
\[
\mathbf{K}_Z=\{K_{z}\}_{{z}\in Z},
\]   
where each $K_{z}$ is a set of independent d.o.f.. 

Let $\Lambda$ be a set which consists of all non-degenerate pairs $\{(\hat{F},K_{{z}})\}$, where ${z}$ runs through $Z$.  We say that $\lambda'=(\hat{F}',K_{{z} '})$ is greater or equal to $\lambda=(\hat{F},K_{{z} })$, $\lambda'\geq\lambda$, if $\hat{F}'\supset\hat{F}$ and ${z}'\geq {z}$. 
  
\begin{pro} {\rm \cite{lqg-tens}}
Suppose that 
\begin{enumerate}
\item \label{gamma>N} for every ${z}\in Z$ and for every natural number $N$ there exists ${z}'\in Z$ such that ${z}'\geq{z} $ and the number of elements of $K_{{z}'}$ is greater than $N$, 
\item \label{lin-comb-Kg} if ${z} '\geq{z} $ then each d.o.f. in $K_{{z}}$ is a cylindrical function compatible with  $K_{{z}'}$,
\item \label{F-Kg-ndeg} for every ${z}\in Z $ there exists $\hat{F}\in\mathbf{F}$ such that $(\hat{F},K_{{z}})\in \Lambda$.
\end{enumerate}
Assume moreover that the set $(\Lambda,\geq)$ satisfies Assumption \ref{k-Lambda}, \ref{comp-f} and \ref{const} (presented in Section \ref{outline}). Then $(\Lambda,\geq)$ is a directed set.
\label{Lambda-dir}     
\end{pro}

Applying this general construction to DPG we choose as a set $(Z ,\geq)$ the directed set $(\tl{\mathbf{L}},\geq)$ of all sets of independent hoops and as a set $\mathbf{K}_Z$ the set $\{K_{\tl{L}}\}_{\tl{L}\in\tl{\mathbf{L}}}$ of all finite sets of configurational d.o.f. defined by these sets of hoops. Consequently, the resulting set $(\La,\geq)$ consists of all non-degenerate pairs $\{(\hat{F},K_{\tl{L}})\}_{\tl{L}\in\tl{\mathbf{L}}}$ and $(\hat{F}',K_{\tl{L}'})\geq(\hat{F},K_{\tl{L}})$ if
\begin{align}
\hat{F}'&\supset\hat{F},& \tl{L}'&\geq\tl{L}.
\label{La-geq}
\end{align}
To convince ourselves that $(\La,\geq)$ is directed it is enough to show that the proposition above is applicable to this case.

Let $L$ be a set of independent loops which defines a set $\tl{L}$ of independent hoops. It is possible to find a set $L_0$ of $N$ independent loops such that the loops in it have no common points with the loops in $L$ \footnote{Consider an analytic disc $D\subset \Sigma$ with boundary such that $(i)$ it has no common points with every $l\in L$ and $(ii)$ its boundary is an analytic submanifold of $\Sigma$. Choose $N$ pairwise distinct points on the boundary of $D$ and a point $y_0$ in the interior of the disc. Next connect each point on the boundary with $y_0$ by an edge contained in $D$. Then a part of the boundary between two consecutive points $y,y'$ and the edges connecting $y_0$ and $y,y'$ form a loop. In this way we obtain a set $L_0$ of $N$ independent loops.}. Then  $L\cup L_0$ is again a set of independent loops which defines a set $\tl{L}'$ of independent hoops. It is clear that  $\tl{L'}\geq \tl{L}$ and that $\tl{L}'$ contains more than $N$ elements. Thus the first assumption of the proposition is satisfied.                     

Regarding the second assumption: by virtue of Proposition \ref{L'L-lin} if $\tl{L}'\geq \tl{L}$ then each d.o.f.  in $K_{\tl{L}}$ is a linear combination of d.o.f. in $K_{\tl{L}'}$ and thereby a cylindrical function compatible with the latter set.   
 
Let $\{l_1,\ldots,l_N\}$ be a set of independent loops and let $e_I$ be an edge which appears in a composition defining the loop $l_I$ and which is intersected by other edges in the composition and by other loops  at most at a finite number of points. It is possible to find a face $S_I$ such that $S_I\cap l_J$ is empty if $I\neq J$ and consists of exactly one point belonging to the interior of $e_I$ if $I=J$. The orientation of $S_I$ can be chosen in such a way that $\hat{\varphi}_{S_I}\kappa_{\tll_J}=\delta_{IJ}$. Then $\hat{F}:={\rm span}\{\hat{\varphi}_{S_1},\ldots,\hat{\varphi}_{S_N}\}$ and $K_{\tl{L}}$, where $\tl{L}=\{\tll_1,\ldots,\tll_N\}$ form a non-degenerate pair being an element of $\La$. This shows that the third assumption is met.

Due to this fact and Lemma \ref{K0-lin-KL} the set $(\La,\geq)$ satisfies Assumption \ref{k-Lambda}. It also meets Assumptions \ref{comp-f} and \ref{const} by virtue of, respectively, Equation \eqref{hphi-Psi} and Equation \eqref{hphiS-kl}.

Thus Proposition \ref{Lambda-dir} guarantees that the set $(\La,\geq)$ built (according to the general construction presented just before the proposition) from the sets $(\tl{\mathbf{L}},\geq)$ and $\{K_{\tl{L}}\}_{\tl{L}\in\tl{\mathbf{L}}}$ is directed.   

%***************************************************
\subsection{Checking Assumptions}
%***************************************************

Since  we have already the directed set $(\La,\geq)$ of (constrained) finite physical systems the only task to be done is to check whether it is constructed {\em properly}, that is, whether a family of spaces of quantum states built over the set  can be equipped with a structure of a projective family. To fulfill the task it is enough to check whether the set $(\La,\geq)$ satisfies all Assumptions listed in  Section \ref{outline}.   

We know already that Assumptions \ref{k-Lambda}, \ref{comp-f} and \ref{const} are met by the set. 

Regarding Assumption \ref{f-Lambda}, consider a set $F_0=\{\varphi_{S_1},\ldots,\varphi_{S_N}\}$ of elementary momentum d.o.f.. Let us fix a face $S_I$ and recall that by virtue of the definition of a face and Equation \eqref{eps} there exists a non-trivial hoop $\tll$ such that $\eps(S_I,\tll)\neq 0$ which means that $\hat{\varphi}_{S_I}\kappa_{\tll}\neq 0$. Due to Lemma \ref{h-h} there exists a set $\tl{L}=\{\tll_1,\ldots,\tll_n\}$ of independent hoops such that $\tll$ is a composition of the hoops in $\tl{L}$ and their inverses. Assume that the hoops are ordered in such a way that $\hat{\varphi}_{S_I}\kappa_{\tll_1}\neq 0$ (if $\hat{\varphi}_{S_I}\kappa_{\tll_i}=0$ for every $i$ then $\hat{\varphi}_{S_I}\kappa_{\tll}= 0$). As shown in Section \ref{dir-set} there exists a set $\{S'_1,\ldots,S'_n\}$ of faces such that $\hat{\varphi}_{S'_i}\kappa_{\tll_j}=\delta_{ij}$. Let 
\[
\hat{F}_I={\rm span}\{\hat{\varphi}_{S_I},\hat{\varphi}_{S'_2},\ldots,\hat{\varphi}_{S'_n}\}.
\]
Then $\la_I:=(\hat{F}_I,K_{\tl{L}})$ is non-degenerate and thereby an element of $\La$. Since $\La$ is directed there exists $\la\in\La$ such that $\la\geq\la_I$ for all $I\in\{1,\ldots,N\}$. Taking into account the definition \eqref{La-geq} of the directing relation on $\La$ we see that all operators $\{\hat{\varphi}_{S_1},\ldots,\hat{\varphi}_{S_N}\}$ belong to a space $\hat{F}\in\la$. Thus the set $(\La,\geq)$ satisfies Assumption \ref{f-Lambda}.              
                 
Assumption \ref{RN} is satisfied by virtue of Lemma \ref{K-bij}, and Assumption \ref{non-deg} by virtue of the construction of the set $(\La,\geq)$ (see the previous section). 

To show that Assumption \ref{Q'=Q} is met we will use 

\begin{pro}
{\rm \cite{q-nonl}} Let $K,K'$ be sets of independent d.o.f. of $N$ and $N'$ elements respectively such that $Q_K=Q_{K'}$. Suppose that there exists a set $K''$ of independent d.o.f. of $N''$ elements such that the image of $\tilde{{K}}''$ is $\R^{{N}''}$ and  each d.o.f. in $K\cup K'$ is a linear combination of d.o.f. in ${K}''$. Then each d.o.f. in $K$ is a linear combination of d.o.f. in $K'$.  
\label{KK'K''}
\end{pro}
\noindent Assume then that $(\hat{F},K_{\tl{L}'}),(\hat{F},K_{\tl{L}})\in\Lambda$ and $Q_{K_{\tl{L}'}}=Q_{K_{\tl{L}}}$. There exists a set $\tl{L}''$ of independent hoops such that $\tl{L}''\geq\tl{L}',\tl{L}$. By virtue of Lemma \ref{K-bij} the image of $\tl{K}_{\tl{L}''}$ is $\R^{N''}$ with $N''$ being the number of elements of $K_{\tl{L}''}$. Proposition \ref{L'L-lin} implies that every d.o.f. in  $K_{\tl{L}}\cup K_{\tl{L}'}$ is a linear combination of d.o.f. in $K_{\tl{L}''}$. Therefore Proposition \ref{KK'K''} allows us to conclude that each d.o.f. in $K_{\tl{L}}$ is a linear combination of d.o.f. in $K_{\tl{L}'}$ which means (Proposition \ref{L'L-lin} again) that $\tl{L}'\geq \tl{L}$. Thus $(\hat{F},K_{\tl{L}'})\geq(\hat{F},K_{\tl{L}})$ and Assumption \ref{Q'=Q} is satisfied. 

Let us finally note that the definition \eqref{La-geq} of the directing relation $\geq$ on $\La$  and Proposition \ref{L'L-lin} imply straightforwardly that the set $(\La,\geq)$ meets Assumptions \ref{lin-comb} and \ref{FF'}. 

We conclude that the set $(\La,\geq)$ satisfies all Assumptions and (as stated in Section \ref{outline}) this fact guarantees what follows. Let $\la=(\hat{F},K_{\tl{L}})\in\La$ and let $d\mu_\la$ be a Lebesgue measure on $\tl{Q}_{K_{\tl{L}}}$ given by the coordinate frame \eqref{K-fr}. Define a Hilbert space   
\[
\h_{\la}:=L^2(\tl{Q}_{K_{\tl{L}}},d\mu_\la)
\]
and assign to $\la$ the set $\mathfrak{S}_\la$ of all states on the $C^*$-algebra $\B_\la$ of all bounded operators on $\h_\la$. Then the set $\{\mathfrak{S}_\la\}_{\la\in\La}$ possesses a structure of a projective family. 

Elements of the projective limit $\mathfrak{S}$ of the family $\{\mathfrak{S}_\la\}_{\la\in\La}$ are  {\em constrained projective quantum states for DPG} where the adjective ``constrained'' refers to the fact that the states were constructed over the set $P_{\rm G}\times \tl{Q}$ of orbits of the gauge transformations given by the Gauss constraint $C_{\rm G}$ which pass through solutions of the constraint.

%***************************************************
\section{Discussion \label{disc}}
%***************************************************

Let us recall that the goal of the present research is to contribute to establishing a general prescription for constructing constrained projective quantum states for field theories. Taking into account the three steps of the scheme of constructing unconstrained states  outlined in Section \ref{intro} one can imagine the following  sorts of procedures for constructing constrained ones:
\begin{enumerate}
\item one solves (classical) constraints on a phase space of a theory and builds constrained projective quantum states for the theory over a set of these solutions quotiented by gauge transformations;  
\item one chooses a set $\bar{\La}$ of unconstrained finite physical systems and on the basis of the constraints on the phase space one transform each system into a constrained one and ``quantizes'' it; thus one obtains constrained quantum states for all finite systems and then combines the states into projective ones for the theory \cite{proj-lt-I}; 
\item one chooses a set $\bar{\La}$ of unconstrained finite physical systems, ``quantizes'' each system and on the basis of the constraints of the theory obtains constrained quantum states for each system, then one uses all such states to build constrained projective states for the theory; 
\item one constructs a space $\bar{\mathfrak{S}}$ of unconstrained projective quantum states and solves quantum constraints defined on this space. 
\end{enumerate}

Clearly, in this paper we applied a procedure of the sort 1. It has to be emphasized that this was possible only because the Gauss constraint of DPG is very simple which allowed to find easily gauge invariant elementary d.o.f. and to describe explicitly all orbits of the gauge transformations generated by the constraint. For obvious reasons a procedure of this sort cannot be applied to many physically important theories. On the other hand procedures of the sort 4 do not seem to be workable because spaces of unconstrained projective quantum states are usually very complex. It seems then that the only promising procedures are those of the sort 2 and 3.

Since the unconstrained and constrained projective quantum states for DPG were constructed independently we can try to relate finite physical systems used in both constructions or to relate quantum states associated with the finite systems in both cases in order to get some hints which may be helpful in formulating a procedure of the sort 2 or 3. It turns out that it is easier and more natural to find a relation between the finite systems. Before that let us describe briefly the construction of the unconstrained states for DPG. 

%***************************************************
\subsection{Unconstrained projective quantum states for DPG}
%***************************************************

Unconstrained projective quantum states for DPG were constructed in \cite{q-nonl} according to the prescription which was used in the present paper to construct the constrained quantum states. The former construction is also background independent (it also uses some LQG techniques like graphs \cite{al-hoop,baez,top-lewand}) and differs from the latter one in details listed below. Before we will describe the differences let us emphasize that many symbols which will refer to the construction of the unconstrained states will be equipped with a bar at the top of them unless a symbol cannot be confused with an analogous one used in the construction of the constrained states. 

The point of departure for the construction of unconstrained states was the phase space $P\times Q$ of the theory. Configurational elementary d.o.f. were defined as integrals of $A$ along edges, integrals of $\sigma$ over faces were chosen to be momentum configurational d.o.f., and faces were defined as in the present paper except the last requirement in the definition of a face introduced in Section \ref{sec-lf}. Configurational d.o.f. were grouped into finite sets given by graphs in $\Sigma$: a graph is a finite set $\gamma=\{e_1,\ldots,e_N\}$ of edges\footnote{Unlike in the present paper in \cite{q-nonl} the source and the target of an edge are allowed to coincide. However, each such an edge is a composition of two edges with two-point boundaries. Here we exclude from the construction of unconstrained projective quantum states those edges which do not possess two-point boundaries---the exclusion results in no essential change of the construction and will make simpler further considerations.}  such that for $I\neq J$  an intersection $e_I\cap e_J$ is either empty or consists of points which belong to boundaries of the edges. Here  a set $\{\kappa_{e_1},\ldots,\kappa_{e_N}\}$ of d.o.f. given by edges of a graph $\gamma=\{e_1,\ldots,e_N\}$ will be denoted by ${K}_\gamma$.  

A momentum operator $\hat{\bar{\varphi}}_{\bar{S}}$ defined by a face $\bar{S}$ on a complex linear space $\overline{\Cyl}$ spanned by cylindrical functions was defined via the same regularization \cite{acz} and reads
\begin{equation}
\hat{\bar{\varphi}}_{\bar{S}}\bar{\Psi}:=\sum_{I=1}^N\eps(\bar{S},e_I)\, \pr_{{K}_\gamma}^*\big(\partial_{\bar{x}_I}\bar{\psi}\big),
\label{u_hphi_S}
\end{equation}
where all symbols have analogous meaning to those appearing in \eqref{hphi_S}, and a number $\eps(\bar{S},e)$ can be calculated as follows. Express the edge $e$ as a composition of edges adapted to $\bar{S}$ and let   
\begin{enumerate}
\item $n$ be a number of transversal edges in the composition which either $(i)$ intersect $\bar{S}$ only at their sources and are placed 'below' the face  or $(ii)$ intersect $\bar{S}$ only at their targets and are placed 'above' the face;
\item $m$ be a number of transversal edges in the composition which either $(i)$ intersect $\bar{S}$ only at their sources and are placed 'above' the face  or $(ii)$ intersect $\bar{S}$ only at their targets and are placed 'below' the face.
\end{enumerate}
Then $\eps(\bar{S},e)=n-m$. 

The set $\Gamma$ of graphs was equipped with the standard  directing relation applied in LQG ($\gamma'\geq\gamma$ if every edge of $\gamma$ is a composition of edges of $\gamma'$ and their inverses) and used to obtain a set $(\bar{\La},\geq)$ of finite physical systems. More precisely, if $\hat{\bar{F}}$ denotes a finite dimensional linear space of momentum operators then $(i)$ $\bar{\La}$ is the set of all non-degenerate pairs $\{(\hat{\bar{F}},{K}_\gamma)\}$ and $(ii)$  $(\hat{\bar{F}}',{K}_{\gamma'})\geq(\hat{\bar{F}},{K}_\gamma)$  if $\hat{\bar{F}}'\supset \hat{\bar{F}}$ and $\gamma'\geq\gamma$. The set $(\bar{\La},\geq)$ satisfies all Assumptions listed in Section \ref{outline} which means that it generates a space $\bar{\mathfrak{S}}$ of projective quantum states.  

%***************************************************
\subsection{The unconstrained quantum states for DPG versus the constrained ones}
%***************************************************

In the present section we will show how to obtain the constrained finite systems for DPG being elements of the set $\La$ from the unconstrained ones constituting the set $\bar{\La}$. 

%***************************************************
\subsubsection{Elementary classical variables of a finite system \label{ecv}}
%***************************************************

Consider a finite system $\bar{\la}=(\hat{\bar{F}},{K}_\gamma)\in\bar{\La}$ and a set $\overline{\Cyl}_{K_\gamma}$ of cylindrical functions compatible with the set $K_\gamma$ (see Equation \eqref{cyl-f}). Let us recall that in canonical quantization a term {\em elementary classical variable} means a function on a phase space which is unambiguously represented by an operator on a Hilbert space. A crucial observation is that (some) functions in $\overline{\Cyl}_{K_\gamma}$ and operators in $\hat{\bar{F}}$ can be treated as elementary classical variables of the system $\bar{\la}$. This statement can be justified as follows.

Let $\gamma=\{e_1,\ldots,e_N\}$. Then the  image of the map $\tilde{K}_\gamma$ (Equation \eqref{K-fr}) is $\R^N$ \cite{q-nonl}. Thus the reduced configuration space $Q_{K_\gamma}$ is an $N$-dimensional linear space and the natural coordinate frame $(\bar{x}_1,\ldots,\bar{x}_N)$ given by $\tilde{K}_\gamma$ defines on it a Lebesgue measure $d\mu_{\bar{\la}}$. This measure yields a Hilbert space   
\[
\h_{\bar{\la}}:=L^2(Q_{K_\gamma},d\mu_{\bar{\la}})
\]   
assigned to the system $\bar{\la}$. 

Note that due to Equation \eqref{cyl-f} there is one-to-one correspondence between elements of $\overline{\Cyl}_{K_\gamma}$ and elements of $C^\infty(Q_{K_\gamma},\C)$ (i.e. smooth complex functions on $Q_{K_\gamma}$) and (some\footnote{These functions can be specified as e.g. multipliers of Schwarz functions on $Q_{K_\gamma}$: $\psi\in C^\infty(Q_{K_\gamma},\C)$ is such a multiplier if for every Schwarz function $\psi'$ on $Q_{K_\gamma}$ $\psi\psi'$ is still a Schwarz function.}) functions in $C^\infty(Q_{K_\gamma},\C)$ define by multiplication (unbounded and bounded) operators on the Hilbert space. 

On the other hand since $\bar{\la}$ satisfies Assumptions \ref{comp-f} and \ref{const} \cite{q-nonl} every $\hat{\bar{\varphi}}\in\hat{\bar{F}}$ defines on $Q_{K_\gamma}$ a constant vector field
\begin{equation}
\sum_{I=1}^N(\hat{\bar{\varphi}}\kappa_{e_I})\partial_{\bar{x}_I}
\label{vec}
\end{equation}
and $\hat{\bar{\varphi}}$ is naturally represented on $\h_{\bar{\la}}$ by an (unbounded quantum) momentum operator defined as $-i$ times \eqref{vec}. 

Thus we can treat functions in $\overline{\Cyl}_{K_\gamma}$ and operators in $\hat{\bar{F}}$ as elementary classical variables of the system $(\hat{\bar{F}},K_\gamma)$. Similarly, given constrained system $\la=(\hat{F},K_{\tl{L}})$ functions in ${\Cyl}_{K_{\tl{L}}}$ and operators in $\hat{{F}}$ can be regarded as elementary classical variables of $\la$.    

We will show below that restricting elementary classical variables of an unconstrained finite system to gauge invariant ones one obtains  elementary classical variables of some constrained systems in $\La$ and thus the systems themselves.

%***************************************************
\subsubsection{Gauge transformations on unconstrained finite systems}
%***************************************************

Let $f$ be a function on $\Sigma$. We will use a symbol $g_f$ to denote all sorts of transformations induced by a gauge transformation $A\mapsto A+df\equiv g_f(A)$ generated by the Gauss constraint.  

If ${\kappa}_e$ is a configurational elementary d.o.f. given by an edge $e$ then
\begin{equation}
(g^*_f {\kappa}_e)(A):={\kappa}_e(g_f(A))={\kappa}_e(A+df)={\kappa}_e(A)+{\kappa}_e(df).
\label{gf-kappa}
\end{equation}
Since $f$ is fixed (i.e. independent of $A$) it follows from the result above that level sets of ${\kappa}_e$ coincide with those of $g^*_f{\kappa}_e$. This means that the Gauss constraint generates gauge transformations on every reduced configuration space $Q_{{K}_\gamma}$:
\begin{equation}
Q_{{K}_\gamma}\ni[A]\mapsto g_f[A]:=[g_f(A)]=[A+df]=[A]+[df]\in Q_{{K}_\gamma},
\label{gf-QK}
\end{equation}
where in the last step we used the natural linear structure on $Q_{{K}_\gamma}$. 

The gauge transformations given by the Gauss constraint preserves the momentum variable $\sigma$, hence every d.o.f. $\bar{\varphi}_{\bar{S}}$ is gauge invariant. But it is not obvious that every momentum $\hat{\bar{\varphi}}_{\bar{S}}$ is gauge invariant since it is defined by means of the regularization which may spoil the invariance. It turns out however (see Appendix \ref{g-inv}) that for every $\bar{\Psi}\in\overline{\Cyl}$  
\[
g^{-1*}_f\big(\hat{\bar{\varphi}}_{\bar{S}} (g^*_f\bar{\Psi})\big)=\hat{\bar{\varphi}}_{\bar{S}}\bar{\Psi}.
\]     

%***************************************************
\subsubsection{Gauge invariant elementary classical variables}
%***************************************************

Let $\bar{\la}=(\hat{\bar{F}},K_\gamma)\in\bar{\La}$ be the system considered in Section \ref{ecv}. If $\bar{\Psi}\in \overline{\Cyl}_{K_\gamma}$ then $\bar{\Psi}=\pr^*_{K_\gamma}\bar{\psi}$ for some $\bar{\psi}\in C^\infty(Q_{K_\gamma},\C)$. Suppose now that $\bar{\Psi}$ is gauge invariant---this means that for every $A\in Q$ and for every function $f$ on $\Sigma$ $\bar{\Psi}(A)=\bar{\Psi}(A+df)$. % or equivalently $\bar{\psi}([A])=\bar{\psi}([A]+[df])$ (see the formula \eqref{gf-QK}).
Using the map $\tilde{K}_\gamma$ we can rewrite this condition as follows:
\begin{equation}
(\bar{\psi}\circ\tl{K}^{-1}_\gamma)\big(\kappa_{e_1}(A),\ldots,\kappa_{e_N}(A)\big)=(\bar{\psi}\circ\tl{K}^{-1}_\gamma)\big(\kappa_{e_1}(A+df),\ldots,\kappa_{e_N}(A+df)\big).
\label{pK-kk}
\end{equation}
Since $\bar{\psi}$ is constant along orbits of the gauge transformations \eqref{gf-QK} it should be possible to express the function $\bar{\psi}\circ\tl{K}^{-1}_\gamma$ as one depending on a lower number of variables which are also constant on the orbits. This can be achieved by means of a prescription \cite{freidel} which once applied to our case uses the gauge transformations to reduce to zero as many variables at the r.h.s. of  \eqref{pK-kk} as possible. 

The prescription gives the following results (see Appendix \ref{gauge-inv} for details): if it is impossible to compose a non-trivial loop from edges of $\gamma$ and their inverses then each gauge invariant cylindrical function compatible with $K_\gamma$ is a constant function on $Q$. Otherwise there exists a set $\{l_1,\ldots,l_M\}$ ($M<N$) of independent loops composed from edges of $\gamma$  and a permutation $\pi$ of an $N$-element sequence such that for every $A\in Q$      
\begin{equation}
(\bar{\psi}\circ\tl{K}^{-1}_\gamma)\big(\kappa_{e_1}(A),\ldots,\kappa_{e_N}(A)\big)=(\bar{\psi}\circ\tl{K}^{-1}_\gamma)\Big(\pi\big(\kappa_{l_1}(A),\ldots,\kappa_{l_M}(A),0,\ldots,0\big)\Big).
\label{pK-kk-l}
\end{equation}
Thus the set of all gauge invariant functions in $\overline{\Cyl}_{K_\gamma}$ is either trivial (i.e. consists of all constant functions) or coincides with a set $\Cyl_{K_{\tl{L}}}$ of cylindrical functions compatible with a set $K_{\tl{L}}$ of configurational elementary d.o.f. on $\tl{Q}$, where the set $\tl{L}$ of independent hoops is defined by the loops $\{l_1,\ldots,l_M\}$.  

We know already that every momentum operator $\hat{\bar{\varphi}}_{\bar{S}}$ is gauge invariant. Consequently, every element of $\hat{\bar{F}}$ is gauge invariant as well. Moreover, it is straightforwardly to check that an operator $\hat{\bar{\varphi}}_{\bar{S}}$ given by \eqref{u_hphi_S} once restricted to gauge invariant functions constituting the space  $\Cyl$ coincides with an operator $\hat{\varphi}_{\bar{S}}$ given by \eqref{hphi_S}.     

It is then tempting to recognize gauge invariant functions in $\overline{\Cyl}_{K_\gamma}$ (provided they are non-trivial) and operators in $\hat{\bar{F}}$ restricted to $\Cyl$ as elementary classical variable of a constrained finite physical system  and the pair $(\hat{\bar{F}},K_{\tl{L}})$ as the constrained system itself. However such a conclusion would be premature, because the pair is not non-degenerate.

To see this define 
\[
\hat{\bar{F}}_0:=\{\ \hat{\bar{\varphi}}\in \hat{\bar{F}}\ | \ \hat{\bar{\varphi}}\kappa_{\tll}=0\ \text{for all $\tll\in\tl{L}$} \ \}.
\]
Obviously, $\hat{\bar{F}}_0$ is linear subspace of $\hat{\bar{F}}$. Since each loop $l_I$ introduced above is a composition of edges of $\gamma$ and their inverses each function $[A]\mapsto \kappa_{\tll_I}(A)$ is a linear function on $Q_{K_\gamma}$. Moreover, such functions given by all the loops $\{l_1,\ldots,l_M\}$ are linearly independent (because the loops are independent). On the other hand we know (see Section \ref{ecv}) that every momentum operator $\hat{\bar{\varphi}}\in\hat{\bar{F}}$ defines a constant vector field on $Q_{K_\gamma}$. This fact together with the non-degeneracy of $(\hat{\bar{F}},K_\gamma)$ (see Section \eqref{dir-set}) provides us with a linear bijection between $\hat{\bar{F}}$ and a linear space of all constant vector fields on $Q_{K_\gamma}$. Thus elements of $\hat{\bar{F}}_0$ correspond to constant vector fields on the $N$-dimensional linear space $Q_{K_\gamma}$ which annihilate $M$ linearly independent linear functions on the space. This means that $\hat{\bar{F}}_0$ is an  $(N-M)$-dimensional linear subspace of $\hat{\bar{F}}$, where $N-M>0$.   

It is clear now that in order to obtain a constrained finite system from $(\hat{\bar{F}},K_\gamma)$ we should choose a linear subspace $\hat{\bar{F}}_1\subset \hat{\bar{F}}$ such that
\[
\hat{\bar{F}}=\hat{\bar{F}}_1\oplus\hat{\bar{F}}_0,
\]
and restrict operators in $\hat{\bar{F}}_1$ to $\Cyl$. If $\hat{{F}}_1$ denotes the resulting space of restricted operators then $(\hat{{F}}_1,K_{\tl{L}})$ is non-degenerate and therefore it is a constrained physical system, that is, an element of $\La$. Note however, that $\hat{\bar{F}}_1$ is not unique---in fact there are infinitely many such subspaces of $\hat{\bar{F}}$ and each of them defines together with $K_{\tl{L}}$ a distinct element of $\La$. 

Thus if there are non-trivial gauge invariant functions in $\overline{\Cyl}_{K_\gamma}$ then the unconstrained finite system $\bar{\la}=(\hat{\bar{F}},K_\gamma)\in\bar{\La}$ defines infinitely many distinct constrained finite system being elements of $\La$. However, we do not claim that every constrained finite system in $\La$ can be obtained from an unconstrained one in $\bar{\La}$ (we will argue below this issue is not very relevant).

Let us emphasize that if an unconstrained system $\bar{\la}$ is a subsystem of an unconstrained one $\bar{\la}'$, that is, $\bar{\la}'\geq \bar{\la}$, and a constrained system $\la'$ ($\la$) is obtained from $\bar{\la}'$ ($\bar{\la}$) then it does not have to mean that $\la'\geq\la$. Indeed, suppose that edges $\{e_1,e_2,e_3\}$ constituting a graph $\gamma'$ have a common target and a common source and that faces $\{S_1,S_2,S_3\}$ are chosen in such a way that $\hat{\bar{\varphi}}_{S_I}\kappa_{e_J}=\delta_{IJ}$ \cite{q-nonl}. Then the graph $\gamma'$ and the faces  $\{S_1,S_2,S_3\}$ define in an obvious way a system $\bar{\la}'\in \bar{\La}$, and a graph $\gamma=\{e_1,e_2\}$ and the faces $\{S_1,S_2\}$ do a system $\bar{\la}\in\bar{\La}$. Obviously $\bar{\la}'\geq\bar{\la}$. Applying the procedure described above we can obtain from $\bar{\la}'$ a constrained system $\la'$ given by loops $\{e^{-1}_2\circ e_1, e^{-1}_3\circ e_1\}$  and faces $\{S_2,S_3\}$, and from $\bar{\la}$ a constrained system $\la$ generated by the loop $e^{-1}_2\circ e_1$ and the face $S_1$. Since $\hat{\varphi}_{S_1}$ do not belong to ${\rm span}\{\hat{\varphi}_{S_2},\hat{\varphi}_{S_3}\}$ $\la$ is not a subsystem of $\la'$. But of course if $\la$ was defined by the same loop and  $\hat{\varphi}_{S_2}$ then it would be a subsystem of $\la'$.                      

To summarize the results just obtained let us assume that we are trying to get constrained projective quantum states for DPG by the procedure described above, that is, by restricting elementary classical variables of each unconstrained finite system in $\bar{\La}$ to gauge invariant ones. Even if we did not obtain in this way the whole set $\La$ of constrained finite systems the procedure gives us an unambiguous answer to the question what a constrained finite system is and if some constrained systems besides those obtained directly from unconstrained ones were needed  they can be easily introduced. Similarly, the procedure does not transform the directing relation on $\bar{\La}$ into the directing relation on $\La$ but it provides us with a clear hint how to define the latter one: if $\gamma'$ is a graph underlying an unconstrained  system $\bar{\la}'$ and $\gamma$ a graph underlying its subsystem $\bar{\la}$ then edges of $\gamma$ are compositions of edges of $\gamma'$ and their inverses. Since the procedure suggests using loops or hoops instead of edges then the relation $\la'\geq\la$ between constrained system should involve composing hoops underlying the subsystem $\la$ from hoops (and their inverses) underlying the system $\la'$.    

Thus it seems that in general one should not expect that a procedure will transform a directed set of unconstrained finite physical system precisely into a directed set of constrained ones. It is rather sufficient to obtain a generic or typical constrained system and some (technical) hints how to define a directing relation between constrained systems (since the principle that this relation should be a relation system--subsystem may be too general to be applied in practice).

%***************************************************
\subsection{Procedure for constructing constrained projective quantum states presented in \cite{proj-lt-I}}
%***************************************************

As mentioned in Section \ref{intro} a procedure for constructing constrained projective quantum states was partially elaborated in \cite{proj-lt-I} and applied successfully to a toy-model in \cite{proj-lt-III}. According to our classification introduced at the very beginning of Section \ref{disc} this procedure is of the sort 2 since it starts by defining and solving constraints on each unconstrained finite system before it is ``quantized''. Here we are going to compare this procedure with that formulated in the previous section. 

The first difference between the procedures is a different way of defining finite physical system: in \cite{proj-lt-I} a finite physical system is identified with a finite dimensional phase space, while here such a system is a pair which consists of a finite dimensional configuration space and a finite dimensional space of linear operators which act on functions defined on the configuration space. However, this difference is rather of a technical character and we will not pay much attention to it, since what is really relevant to us are details of a passage from unconstrained systems to constrained ones.    

In \cite{proj-lt-I} it is assumed that constraints on a phase space of a field theory induce (in a strict or in an approximate way) constraints on an unconstrained finite system $\bar{\la}$---these constraints on $\bar{\la}$  will be called here {\em reduced constraints}. Let $\bar{\la}_0\subset\bar{\la}$ be the set of all solutions of the reduced constraints. Those reduced constraints which are of the first class define gauge transformations on $\bar{\la}$. The set of all orbits of these transformations contained in $\bar{\la}_0$ equipped with a symplectic form induced by that on $\bar{\la}$ is recognized as a constrained finite system $\la$. This implies that the procedure associates with every unconstrained system precisely one constrained one. Moreover, the procedure requires a constrained system $\la$ to be a subsystem of a constrained system $\la'$ if $\la$ and $ \la'$ originate respectively from unconstrained systems $\bar{\la}$ and $\bar{\la}'$   such that $\bar{\la}$ is a subsystem of $\bar{\la}'$.

Clearly, the procedure elaborated here in the case of DPG differs significantly from that introduced in \cite{proj-lt-I}: the former procedure does not use reduced constraints at all, it associates with an unconstrained system many constrained ones and a constrained system $\la$ does not have to be a subsystem of a constrained one $\la'$ even if initial unconstrained systems are in such a relation. Moreover, we will show below that at least in a case of some unconstrained systems of DPG there is no basis for introducing any reduced constraints.    

To this end consider an unconstrained system $\bar{\la}$ of DPG generated by a graph $\gamma=\{e_1,e_2\}$ and faces $\{S_1,S_2\}$ such that $(i)$ $l=e_2\circ e_1$ is a loop, $(ii)$ boundaries $\{\partial S_1,\partial S_2\}$ are (smooth) circles, $(iii)$ an intersection $cl(S_1)\cap cl(S_2)$ is empty and $(iv)$ $\hat{\bar{\varphi}}_{S_I}\kappa_{e_J}=\delta_{IJ}$.          

Since the Gauss constraint $C_{\rm G}$ does not impose any restriction on the configuration variable $A$ it does not restrict values of $\kappa_{e_1}$ and $\kappa_{e_2}$.  

Regarding momentum d.o.f. $\bar{\varphi}_{S_1}$ and $\bar{\varphi}_{S_2}$: it is possible to find two one-forms $\alpha_1$ and $\alpha_2$ such that $d\alpha_1\neq 0\neq d\alpha_2$ and $\int_{\partial S_I}\alpha_J=\delta_{IJ}$. Let
\[
\sigma_{t_1t_2}:=t_1d\alpha_1+t_2d\alpha_2,
\]     
where $t_1,t_2$ are real numbers. Clearly, $\sigma_{t_1t_2}$ satisfies the Gauss constraint and
\[
\bar{\varphi}_{S_I}(\sigma_{t_1t_2})=t_I,
\] 
which means that $C_{\rm G}$  does not impose any restriction on pairs consisting of values of $\bar{\varphi}_{S_1}$ and $\bar{\varphi}_{S_2}$. On the other hand it is easy to see that the momentum operators $\hat{\bar{\varphi}}_{S_1},\hat{\bar{\varphi}}_{S_2}$  remain linearly independent once restricted to the space $\Cyl$ spanned by gauge invariant cylindrical functions. 

Thus we see that there are no constraints induced by the Gauss constraint on the system $\bar{\la}$. Despite this fact $C_{\rm G}$ defines gauge transformations on $\bar{\la}$ which allow us to obtain from $\bar{\la}$ many constrained finite systems like that given by the loop $l$ and the face $S_1$.       

The main conclusion which can be drawn from the considerations above is that constraints on a phase space of a field theory may induce gauge transformations on an unconstrained finite system in two ways: via gauge transformations defined on the phase space and via reduced constraints induced on the system i.e. in the former case the constraints first define on the phase space the transformations which induce then transformations on the system and in the latter case the constraints first define the reduced ones which define then transformations on the system. However, the transformations on the system obtained in these two ways do not have to coincide. 

Consequently, to be more general the procedure described in \cite{proj-lt-I} should take into account that except gauge transformations defined on an unconstrained system via reduced constraints there may exist distinct gauge transformations obtained in the other way. Moreover, these additional transformations may spoil the one-to-one correspondence between unconstrained and constrained systems and a simple correspondence between the relations system-subsystem among unconstrained and constrained systems.

%***************************************************
\section{Summary}
%***************************************************

In this paper we constructed constrained projective quantum states for the degenerate Pleba\'nski gravity which satisfy the Gauss constraint. To obtain these states we considered the set $P_{\rm G}\times\tl{Q}$ of all solutions of the constraint quotiented by the gauge transformations generated by the constraint. Then we used this set as the point of departure for the method of constructing projective quantum states introduced in \cite{q-nonl}.

Next we compared the constrained states obtained here with the unconstrained projective quantum states for DPG described in \cite{q-nonl}. The goal of this comparison was to formulate a procedure which would allow us to obtain constrained projective quantum states without refering to the space of classical solutions of the constraint modulo the gauge transformations. Such a procedure was formulated---it transforms unconstrained finite physical systems to constrained ones by restricting elementary classical variables of an unconstrained system to gauge invariant ones. 

Finally we showed that the fairly general procedure aimed at an analogous goal introduced in \cite{proj-lt-I} does not include the one elaborated in this paper. The reason for this is that the former procedure assumes that all gauge transformations on an unconstrained finite system can be generated by reduced constrains induced on the system by constraints on a phase space of a field theory and this assumption is not satisfied in the case of the unconstrained systems chosen in \cite{q-nonl} for DPG.    

Certainly, more examples of constrained projective quantum states should be constructed and analyzed in order to formulate a general procedure for constructing such states.  

\paragraph{Acknowledgments} I am very grateful to Prof. Jacek Jezierski and Wojciech Kami\'n\-ski for helpful hints.

%***************************************************
%***************************************************
\appendix
%***************************************************
%***************************************************

%***************************************************
\section{Proof of Lemma \ref{k-akak} \label{proof}} 
%***************************************************

Let $\tilde{L}=\{\tilde{l}_1,\ldots,\tilde{l}_N\}$. Assume that
\begin{equation}
\kappa_{\tilde{l}}=\sum_{I=1}^{N} a_I\kappa_{\tilde{l}_I},
\label{k-lin-kk}
\end{equation}
where  $\{a_1,\ldots,a_N\}$ are real numbers. Let $l$ be a loop in $\tll$, $\{{l}_1,\ldots,{l}_N\}$ be a set of independent loops defining the set $\tilde{L}$ and 
\begin{align}
l&=e_n\circ\ldots \circ e_1,\label{l-eee}\\
{l}_1&=e_{1m}\circ \ldots \circ e_{1i}\circ\ldots\circ e_{11},\nonumber
\end{align}
where $e_{1i}$ is that edge which is intersected by the remaining edges $\{e_{1j}\}_{j\neq i}$ and by the remaining loops $\{l_2,\ldots,l_N\}$ at most at a finite number of points.

Let us now construct a sequence of edges $(E_0,E_1,\ldots, E_n)$ according to the following prescription: $E_0=e_{1i}$ and  
\begin{enumerate}
\item if $E_{k-1}\cap e_k$ is empty then $E_k=E_{k-1}$,  
\item if $E_{k-1}\cap e_k$ is finite then $E_k$ is defined as a subedge of $E_{k-1}$ such that $E_k\cap e_k$ is empty,   
\item if $E_{k-1}\cap e_k$ is infinite then $E_k$ is defined as a subedge\footnote{This subedge exists because the edges $E_{k-1}$ and $e_k$ are piecewise analytic \cite{al-hoop}.} of $E_{k-1}$ such that $E_k\subset e_k$.
\end{enumerate} 
Note that $E_n$ is an edge such that for every edge $e_i$ in \eqref{l-eee} either $E_n\cap e_i$ is empty or $E_n\subset e_i$ and in the latter case either the orientation of $E_n$ coincides with that of $e_i$ or the orientations are opposite.

Now let us choose a one-form $A$ such that $(i)$ $\int_{E_n}A\neq 0$, $(ii)$ $\supp A\cap l\subset E_n$, $(iii)$  $\supp A\cap l_1\subset E_n$ and $(iv)$ for every remaining loop $\supp A\cap l_I$ is empty. Then
by virtue of \eqref{k-lin-kk}
\[
(j_+-j_-)\int_{E_n}A=a_1\int_{E_n}A,
\]
where $j_+$ is the number of factors\footnote{In this counting all factors/edges in \eqref{l-eee} are treated as distinct if only they appear at different positions in the composition even if some of them coincide as oriented submanifolds of $\Sigma$.} in \eqref{l-eee} which contain $E_n$ and are oriented according to the orientation of $E_n$ and $j_-$  is the number of factors in \eqref{l-eee} which contain $E_n$ and are oriented in the opposite way.

The equation above allows us to conclude that $a_1$ is an integer. Thus all $\{a_1,\ldots,a_N\}$ in \eqref{k-lin-kk} are integers.

Note now that by virtue of \eqref{kkk}
\[
\kappa_{\tilde{l}_1^{a_1}\circ\ldots \circ\tilde{l}_N^{a_N}}=\sum_{I=1}^{N} a_I\kappa_{\tilde{l}_I}.
\]
This equation and \eqref{k-lin-kk} imply that
\[
\tilde{l}=\tilde{l}_1^{a_1}\circ\ldots \circ\tilde{l}_N^{a_N},
\]
which means that $\tilde{l}$ is a composition of hoops in $\tilde{L}$ and their inverses. 

%***************************************************
\section{Momentum operators on $\Cyl$ are well defined \label{hphi-ok}}
%***************************************************

Consider a face $S$ and a set $\{l_0,l_1,\ldots,l_N\}$ of (pairwise distinct) loops such that each of them is a composition of edges adapted to $S$. Denote by ${\cal E}$ the set of all factors\footnote{Again, here every two factors/edges $e,e'$  are treated as distinct if the appear at different compositions or at different positions in the same composition even if $e=e'$ as oriented submanifolds of $\Sigma$.} appearing in these compositions and by ${\cal E}^T$ the set of all factors in $\cal E$ which are transversal to $S$. 

Assume that a hoop $\tilde{l}_0$ is a composition of hoops  $\{\tilde{l}_1,\ldots,\tilde{l}_N\}$ and their inverses. The latter requirement means that 
\begin{equation}
\kappa_{{l}_0}=\sum_{I=1}^N n_I \kappa_{{l}_I},
\label{k-nk}
\end{equation}
where $n_I$ are integers.

Let us define on ${\cal E}^T$ an equivalence relation: $e_i\sim e_j$ if $(i)$ $e_i\cap S=e_j\cap S$ and $(ii)$ if there exists an open neighborhood $U$ of the point $e_i\cap S=e_j\cap S$ such that $e_i\cap U=e_j\cap U$. We will call an equivalence class $[e_i]$ of $e_i$ a germ.           

Consider a germ $[e]$ such that the edges in the germ are placed 'above' $S$.  Then there exists an edge $E$ such that its target is $e\cap S$ and it is contained in every edge in $[e]$. There exists a one-form $A\in Q$ such that for every $e'\in{\cal E}$       
\[
\supp A\cap e'=
\begin{cases}
\supp A\cap E &\text{if $e'\in[e]$}\\
\varnothing &\text{otherwise}
\end{cases}
\]
and $\int_EA\neq 0$. Evaluating both sides of \eqref{k-nk} on $A$ we obtain
\begin{equation}
(t^+_0[e]-s^+_0[e])\int_EA=\sum_{I=1}^N n_I(t^+_I[e]-s^+_I[e])\int_EA,
\label{tstse}
\end{equation}
where $t^+_\alpha[e]$ is the number of edges in $[e]$ which appear in the decomposition of $l_\alpha$ and which intersect $S$ at their targets, similarly $s^+_\alpha[e]$ is the number of edges in $[e]$ which appear in the decomposition of $l_\alpha$ and which intersect $S$ at their sources. 

Now, adding Equations \eqref{tstse} obtained for all germs $\{[e]\}$ such that edges in the  germs are placed 'above' the face $S$ we obtain       
\[
t^+_0-s^+_0=\sum_{I=1}^N n_I(t^+_I-s^+_I),
\]
where $t^+_\alpha$ is the number of edges which appear in the decomposition of $l_\alpha$, intersect $S$ at their targets and are placed 'above' $S$ and $s^+_\alpha$ is the number of edges which appear in the decomposition of $l_\alpha$, intersect $S$ at their sources and are placed 'above' $S$. Obviously, analogous formula
\[
t^-_0-s^-_0=\sum_{I=1}^N n_I(t^-_I-s^-_I)
\]
holds for edges placed 'below' the face (hopefully, the meaning of symbols is obvious here). Thus
\begin{equation}
t^+_0-s^+_0-(t^-_0-s^-_0)=\sum_{I=1}^N n_I\big(t^+_I-s^+_I-(t^-_I-s^-_I)\big).
\label{tsts-fin}
\end{equation}

Let us show now that the number \eqref{eps} does not depend on the choice of the loop \eqref{l-ad-ee-1}. Assume that a set $\{l'_0,l'_1\}$ consists of two distinct loops of this sort i.e. the loops are compositions of edges adapted to $S$ and  $\tilde{l}'_0=\tilde{l}'_1=\tilde{l}_I$. Thus in this case Equation \eqref{k-nk} reads $\kappa_{l'_0}=\kappa_{l'_1}$ and we obtain from \eqref{tsts-fin}      
\[
t^+_0-s^+_0-(t^-_0-s^-_0)=t^+_1-s^+_1-(t^-_1-s^-_1),
\]
which proofs that the number \eqref{eps} is independent of the choice of a representative of $\tilde{l}_I$. 

Now we are going to prove that the r.h.s. of \eqref{hphi_S} does not depend on the choice of a set $\tilde{L}$ of independent hoops. Since the set $\tilde{\mathbf{L}}$ of all sets of independent hoops is directed it is enough to show this in the case of sets $\tilde{L}',\tilde{L}\in\tl{\mathbf{L}}$ such that $\tilde{L}'\geq\tl{L}$. 

If $\tilde{L}'=\{\tll'_1,\ldots,\tll'_{N'}\}\geq\tl{L}=\{\tll_1,\ldots,\tll_N\}$ then 
\begin{equation}
\tll_I=(\tll'_{1})^{n_{1I}}\circ \ldots \circ(\tll'_{N'})^{n_{N'I}},
\label{l-lnln}
\end{equation}
where $\{n_{I'I}\}$ are integers. It is possible to show \cite{q-nonl} that in this case there exists a linear projection $\pr_{K_{\tilde{L}}K_{\tilde{L}'}}$ from $\tl{Q}_{K_{\tilde{L}'}}$ onto $\tl{Q}_{K_{\tilde{L}}}$ such that
\begin{equation}
\pr_{K_{\tilde{L}}}=\pr_{K_{\tilde{L}}K_{\tilde{L}'}}\circ \pr_{K_{\tilde{L}'}}.
\label{prprpr}
\end{equation}
The projection expressed in  natural coordinate frames $(x'_{I'})$ and $(x_I)$ on the reduced configuration spaces reads  
\begin{equation}
(x'_{I'})\mapsto (x_I)=\Big(\sum_{I'=1}^{N'} n_{I'I}x'_{I'}\Big).
\label{x'-x}
\end{equation}

Suppose that $\Psi\in\Cyl$ is a cylindrical function compatible with $K_{\tilde{L}}$. Due to \eqref{prprpr} the function is compatible with $K_{\tilde{L}'}$:
\[
\Psi=\pr^*_{K_{\tilde{L}}}\psi=\pr^*_{K_{\tilde{L}'}}(\pr_{K_{\tilde{L}}K_{\tilde{L}'}}^*\psi)=\pr^*_{K_{\tilde{L}'}}\psi',
\]   
where
\[
\psi':=\pr_{K_{\tilde{L}}K_{\tilde{L}'}}^*\psi
\]
is a smooth complex function on $\tilde{Q}_{K_{\tl{L}'}}$. Using \eqref{x'-x} we obtain from the equation above
\[
\partial_{x'_{I'}}\psi'=\sum_{I=1}^N n_{I'I}\pr_{K_{\tilde{L}}K_{\tilde{L}'}}^*(\partial_{x_I}\psi).
\]
On the other hand it follows from \eqref{eps}, \eqref{tsts-fin} and \eqref{l-lnln} that
\[
\eps(S,\tll_{I})=\sum_{I'=1}^{N'} n_{I'I}\eps(S,\tll'_{I'}).
\]
Combining the last two equations and \eqref{prprpr} we get
\begin{multline*}
\sum_{I'=1}^{N'}\eps(S,\tll'_{I'})\pr^*_{K_{\tilde{L}'}}\big(\partial_{x'_{I'}}\psi'\big)=\pr^*_{K_{\tilde{L}'}}\circ \pr_{K_{\tilde{L}}K_{\tilde{L}'}}^*\Big(\sum_{I=1}^{N}\sum_{I'=1}^{N'} n_{I'I}\eps(S,\tll'_{I'})\partial_{x_{I}}\psi\Big)=\\=
\sum_{I=1}^{N}\eps(S,\tll_{I})\pr^*_{K_{\tilde{L}}}\big(\partial_{x_{I}}\psi\big)
\end{multline*}
which shows that the definition of momentum operators does not depend on the choice of a set $\tl{L}$ of independent hoops.

%***************************************************
\section{Momentum operators on $\overline{\Cyl}$ are gauge invariant \label{g-inv}}
%***************************************************

As proven in \cite{q-nonl} for every $\bar{\Psi}\in\overline{\Cyl}$ there exists a graph $\gamma$ such that $\bar{\Psi}$ is a cylindrical function compatible with ${K}_\gamma$ i.e. there exists a smooth function $\bar{\psi}$ on $Q_{{K}_\gamma}$ such that
\[
\bar{\Psi}=\pr_{{K}_\gamma}^*\bar{\psi}.
\]      
Let $\gamma=\{e_1,\ldots,e_N\}$ and let $(\bar{x}_1,\ldots,\bar{x}_N)$ be a natural coordinate frame given by the map ${\tl{K}}_\gamma$ (see \eqref{K-fr}). There exist one-forms $\{A_1,\ldots,A_N\}$ such that
\[
{\kappa}_{e_J}(A_I)=\delta_{IJ}
\]  
for every ${\kappa}_{e_J}\in{K}_\gamma$. By a direct calculation one can easily show that
\[
\frac{d}{dt}\Big|_{t=0}\bar{\Psi}(A+tA_I)=(\pr_{{K}_\gamma}^*\partial_{\bar{x}_I}\bar{\psi})(A).
\] 
Thus 
\[
(\hat{\bar{\varphi}}_{\bar{S}}\bar{\Psi})(A)=\sum_{I=1}^N\eps(\bar{S},e_I)\frac{d}{dt}\Big|_{t=0}\bar{\Psi}(A+tA_I).
\]
Consequently,
\begin{multline*}
\Big(g^{-1*}_f\big(\hat{\bar{\varphi}}_{\bar{S}}(g^*_f\bar{\Psi})\big)\Big)(A)=\big(\hat{\bar{\varphi}}_{\bar{S}}(g^*_f\bar{\Psi})\big)(A-df)=\sum_{I=1}^N\eps(\bar{S},e_I)\frac{d}{dt}\Big|_{t=0}(g^*_f\bar{\Psi})(A-df+tA_I)=\\=\sum_{I=1}^N\eps(\bar{S},e_I)\frac{d}{dt}\Big|_{t=0}\bar{\Psi}(A-df+tA_I+df)=(\hat{\bar{\varphi}}_{\bar{S}}\bar{\Psi})(A).
\end{multline*}

%***************************************************
\section{Proof of Equation \eqref{pK-kk-l} \label{gauge-inv}}
%***************************************************

The proof of Equation \eqref{pK-kk-l} is based on a method presented in \cite{freidel}.

In considerations below  it will be convenient to use a common name for points which belong to boundaries of edges of a graph---a point $y$ will be called a {\em vertex} of a graph if $y$ is a target or a source of at least one edge of the graph.    

Assume first  that the graph $\gamma=\{e_1,\ldots,e_N\}$ is connected i.e. that the set $e_1\cup\ldots\cup e_N$ is connected. Let us now construct a sequence of $(T_1,\ldots,T_V)$ of subgraphs of $\gamma$, where $(V+1)$ is the number of vertices of $\gamma$: 
\begin{enumerate}
\item $T_1:=\{e_1\}$;
\item for $k<V$ choose an edge $e\in\gamma\setminus T_k$ such that precisely one point of $e$ (either its source of its target) is a vertex\footnote{The graph $T_k$ possesses $(k+1)$ vertices and if $k<V$ then there exist at least one vertex of $\gamma$ which is not a vertex of $T_k$. The edge $e$ exists since otherwise $\gamma$ would not be connected.} of $T_k$ and define $T_{k+1}:=T_k\cup\{e\}$. 
\end{enumerate}
It is easy to see that for every $k\in\{1,\ldots,V\}$ $(i)$ $T_k$ is a connected graph and $(ii)$ it is impossible to obtain a non-trivial loop by a composition of edges of $T_k$ and their inverses.  Note that every vertex of $\gamma$ is a vertex of $T_V$. The graph $T_V$ is called a {\em maximal tree} of $\gamma$.         

Fix a vertex $v_0$ of $\gamma$ and renumber the edges of the graph  in such a way that $\gamma\setminus T_V=\{e_1,\ldots,e_M\}$ (where $0\leq M<N$).

Because $T_V$ is a connected graph and it is impossible to compose a non-trivial loop from its edges and their inverses for every vertex $v$ of $\gamma$ which is distinct from $v_0$ there exists a unique path along $T_V$ which connects $v$ with $v_0$. More precisely, there exists a unique composition
\begin{equation}
e(v)=e'_{m}\circ \ldots\circ e'_{1}, 
\label{ev}
\end{equation}
of edges such that $(i)$ for every $i\in\{1,\ldots,m\}$ either $e'_i$ or its inverse belongs to the tree $T_V$, $(ii)$ if $i\neq j$ then $e'_i$ is distinct from $e'_j$ and its inverse and $(iii)$ the vertex $v$ is the source of $e'_1$ and the vertex $v_0$ is the target of $e'_m$. 

Let us fix a one-form $A\in Q$ and assign to each vertex $v$ of $\gamma$ a number $\theta(v)$ according to the following prescription: the number assigned to $v_0$ is $0$. If $v\neq v_0$ then
\[
\theta(v)= \kappa_{e'_1}(A)+\theta(v'),
\]
where $e'_1$ is the edge appearing in the composition \eqref{ev} and $v'$ is the target of $e'_1$. If $f$ is a function on $\Sigma$ such that for every vertex $v$ of $\gamma$ $f(v)=\theta(v)$ then for every edge $e_I\in T_V$
\begin{equation}
\kappa_{e_I}(A+df)=0.
\label{kappa=0}
\end{equation}
Setting this to  Equation \eqref{pK-kk} we obtain 
\begin{multline}
(\bar{\psi}\circ\tl{K}^{-1}_\gamma)\big(\kappa_{e_1}(A),\ldots,\kappa_{e_N}(A)\big)=\\=(\bar{\psi}\circ\tl{K}^{-1}_\gamma)\big(\kappa_{e_1}(A+df),\ldots,\kappa_{e_M}(A+df),0,\ldots,0\big).
\label{pK-kk-00}
\end{multline}

It follows from the equation above that if $T_V=\gamma$ then $M=0$ and every gauge invariant function compatible with $K_\gamma$ must be constant.

Suppose now that $T_V$ is a proper subset of $\gamma$ and fix an edge $e_J\in\gamma\setminus T_V$.  If neither the source $v$ nor the target $v'$ of $e_J$ coincide with $v_0$ then
\[
l_J:=e(v')\circ e\circ \big(e(v)\big)^{-1}
\]
is a non-trivial loop based at $v_0$ (if $v=v_0$ or $v'=v_0$ then we define respectively, $l_J:=e(v)\circ e$ and $l_J:=e\circ \big(e(v)\big)^{-1}$). Then
\[
\kappa_{l_J}(A)=\kappa_{l_J}(A+df)=\kappa_{e_J}(A+df),
\]
where in the second step we used \eqref{kappa=0}. Setting this to \eqref{pK-kk-00} yields an equality 
\[
(\bar{\psi}\circ\tl{K}^{-1}_\gamma)\big(\kappa_{e_1}(A),\ldots,\kappa_{e_N}(A)\big)=(\bar{\psi}\circ\tl{K}^{-1}_\gamma)\big(\kappa_{l_1}(A),\ldots,\kappa_{l_M}(A),0,\ldots,0\big).
%\label{pK-kk-l0}
\]

Note that neither the loops $\{l_1,\ldots,l_M\}$ nor the ordering of edges of $\gamma$ depend on the one-form $A$. Therefore the result just obtained holds for every $A\in Q$. This result  coincide with \eqref{pK-kk-l} modulo the ordering of variables at the r.h.s. of both equations. 

If the graph $\gamma$ is not connected then we divide it into maximal connected subgraphs and repeat the procedure above for each resulting subgraph separately and obtain finally \eqref{pK-kk-l} modulo the ordering of variables at the r.h.s.. 

It is clear that the set of loops $\{l_1,\ldots,l_M\}$ obtained from a graph according to the procedure $(i)$ is a set of independent loops and $(ii)$ is not unique since it depends on the choice of a maximal tree for every maximal connected subgraph of the graph.

%***************************************************
%***************************************************
%\bibliography{bibliography}{}
%\bibliographystyle{oko}
%***************************************************
%***************************************************

\end{document}